\documentclass[runningheads,envcountsame,envcountsect]{llncs}

\newif\ifarxiv
\arxivtrue

\usepackage{xcolor}
\usepackage{amsmath,amssymb}
\usepackage{proof}
\usepackage{mathtools}
\usepackage{stmaryrd}
\usepackage{subcaption}
\usepackage[disable]{todonotes}
\usepackage{listings}
\definecolor{Mulberry}{rgb}{0.77, 0.29, 0.55}
\usepackage{graphicx}
\usepackage{tikz, calc}
\usetikzlibrary{automata}%necessary at least for accepting states
\usetikzlibrary{calc} %compute coordiates in the picture (these expressions sourrounded by dollar symbols)
\usetikzlibrary{backgrounds} %used for the CEGAR figure
\usetikzlibrary{decorations.pathmorphing,automata,petri}
\usetikzlibrary{positioning}
\usetikzlibrary{plotmarks} % needed for eval plot
\usetikzlibrary{arrows}

\usepackage{wrapfig}

\usepackage[misc]{ifsym}

\usepackage[bookmarks,unicode,colorlinks,allcolors=blue]{hyperref}

\ifarxiv
  \makeatletter

  \hypersetup{hidelinks=true}
  \def\orcidID#1{{\href{http://orcid.org/#1}{\protect\raisebox{-1.25pt}{\protect\includegraphics{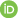}}}}}
  \hypersetup{hidelinks=false,colorlinks,allcolors=blue}
  \makeatother
\fi

\ifarxiv
  \usepackage{apxproof}
  \renewcommand\appendixsectionformat[2]{Additional Material}
\else
  \usepackage[appendix=strip]{apxproof}
\fi

\usepackage[capitalise]{cleveref}

\usepackage{xspace}
\newcommand\true{\mathit{true}}
\newcommand\false{\mathit{false}}

\tikzset{st/.style={font=\ttfamily,shape=rectangle,rounded corners=.5em,fill=gray!40,inner xsep=.3em,inner ysep=0em,text height=2ex,text depth=.6ex}}
\newcommand{\stsmcol}[2]{\tikz[baseline]{\node[st,fill=#2] at (0,.64ex){\hspace{.3em}\texttt{\strut\small #1}\hspace{.3em}\strut};}}
\definecolor{thinit}{RGB}{255, 150, 150}
\definecolor{th1}{RGB}{130, 200, 235}
\definecolor{th2}{RGB}{180, 150, 235}

\newcommand\st{\mathit{st}}

\newcommand\Command{\mathbf{Command}}
\newcommand\CommandPlus{\Command \cup \{\Omega,\lightning\}}

\newcommand\atomicstmt{\mathbf{AtomicStmt}}

\newcommand\prog{\mathcal{P}}

\newcommand\threadTemplates{\mathbf{Templates}}
\newcommand\main{\textup{\texttt{main}}\xspace}

\newcommand\vars{\textup{\textbf{Var}}}
\newcommand\globalVars{\textbf{Globals}}

\newcommand\body[1]{\mathit{body}_{{#1}}}

\newcommand\proglang{\textsc{Conc}\xspace}

\newcommand{\sem}[1]{\llbracket#1\rrbracket}

\newcommand\semtrans[1]{\xrightarrow{#1}}

\newcommand\srule[1]{(\textsc{#1})}

\newcommand\petriProg[2]{\mathit{P\!N}_{\!#1}(#2)}

\newcommand\spec{\mathcal{S}}
\newcommand\safe{\mathsf{safe}}
\newcommand\bound{\mathsf{bound}}

\newcommand\petrirel[1]{\xhookrightarrow{#1}}

\newcommand\insuff[1]{\mathsf{insuff}^{#1}}
\newcommand\inuse[2]{\mathsf{inUse}_{#2}^{#1}}
\newcommand\notinuse[2]{\mathsf{notInUse}_{#2}^{#1}}

\newcommand\instvar[3]{\ensuremath{\texttt{#1}^{#2}_{#3}}}
\newcommand\idvar[2]{\instvar{id}{#1}{#2}}

\newcommand\petristate{\mathbf{State}}

\newcommand{\marki}{\operatorname{mark}}
\newcommand{\insta}{\operatorname{inst}}

\newcommand{\confOp}{\operatorname{conf}}
\newcommand{\conf}[2]{\confOp(#1, #2)}

\newcommand\aug[1]{\bar{#1}}

\newcommand\state[2]{\insta(#1,#2)}

\newcommand\init{\mathit{init}}

\newcommand{\N}{{\mathcal{N}}}
\newcommand{\lab}{\lambda}
\newcommand{\fire}[1]{\vartriangleright_{#1}}
\newcommand\occ{\mathcal{O}}

\newcommand\toolname{\textsc{Ultimate}\xspace}

\usepackage{enumitem}

\newtheoremrep{theorem}{Theorem}[section]
\newtheoremrep{lemma}{Lemma}[section]

\begin{document}
%
%\title{Contribution Title\thanks{Supported by organization x.}}
%\title{Automated Thread Bound Detection and Verification of Concurrent Programs with Dynamic Thread Creation}
\title{Petrification: Software Model Checking for Programs with Dynamic Thread Management \ifarxiv\\(Extended Version)\fi}
\titlerunning{Petrification: Software Model Checking with Dynamic Thread Management}
\author{Matthias Heizmann\textsuperscript{(\Letter)}~\orcidID{0000-0003-4252-3558} \and Dominik~Klumpp~\orcidID{0000-0003-4885-0728} \and Lars Nitzke \and Frank~Sch\"ussele~\orcidID{0000-0002-5656-306X}}

\authorrunning{Matthias Heizmann \and Dominik~Klumpp \and Lars Nitzke \and Frank~Sch\"ussele}
\institute{University of Freiburg, Freiburg im Breisgau, Germany\\
\email{\{heizmann,klumpp,schuessf\}@informatik.uni-freiburg.de}\\
\email{lars.nitzke@mailfence.com}}
\maketitle
\begin{abstract}
We address the verification problem for concurrent program that dynamically create (fork) new threads or destroy (join) existing threads.
We present a reduction
%We reduce the verification problem for concurrent programs that generate threads dynamically
to the verification problem for concurrent programs with a fixed number of threads.
More precisely, we present
\emph{petrification},
a transformation from programs with dynamic thread management to an existing, Petri net-based formalism for programs with a fixed number of threads.
Our approach is implemented in a software model checking tool for C programs that use the \emph{pthreads} API.

%an algorithm whose input is a concurrent program with dynamic thread management together with a safety specification.
%Our algorithm utilizes other algorithms that check this specification for concurrent program with a fixed number of threads and outputs whether the original program satisfies its specification.

\end{abstract}
%
%
%
% !TEX root = ../main.tex

\section{Introduction}
%\todo{Rewrite introduction as proposed by Reviewer A}
%%%%%%%%%%%%%%%%%%%%%%%%%%%%%%%%%%%%%%%%%%%%%%%%%%%%%%%%%%%%%%%%%%%%%%%%%%%%%%%%
% WHAT
%%%%%%%%%%%%%%%%%%%%%%%%%%%%%%%%%%%%%%%%%%%%%%%%%%%%%%%%%%%%%%%%%%%%%%%%%%%%%%%%
%
%\begin{itemize}
%  \item algorithmic verification of concurrent programs with dynamic thread management
%  \item introduce a simple language, \proglang, with fork, join, thread templates
%  \item show example program; verify: assert never violated
%  \item goal: fully sound verification, without *assuming* some bounds
%  \item approach: instead *determine* natural bounds (note: many pgms are bounded)
%  \item reduce to verification w/ limited threads, but check if program can exceed limits
%  \item relatively complete for bounded prgms; can find errors even in unbounded pgms
%\end{itemize}
%
%
%
%%%%%%%%%%%%%%%%%%%%%%%%%%%%%%%%%%%%%%%%%%%%%%%%%%%%%%%%%%%%%%%%%%%%%%%%%%%%%%%%
% What is dynamic thread management?
% Why is verification of programs with dynamic thread management important?
%%%%%%%%%%%%%%%%%%%%%%%%%%%%%%%%%%%%%%%%%%%%%%%%%%%%%%%%%%%%%%%%%%%%%%%%%%%%%%%%
%
We address the verification problem for concurrent programs with \emph{dynamic thread management}.
Such programs start with a single main thread,
and dynamically create (\emph{fork}) and destroy (\emph{join}) additional concurrently executing threads.
%In such programs,
%Dynamic thread management refers to the fact that
%the concurrent threads
%of a program
%are not specified declaratively, or managed by some high-level abstractions,
%but must be explicitly created (\emph{forked}) and destroyed (\emph{joined}) by the program.
%
%
\begin{wrapfigure}[16]{r}{0.33\textwidth}
\vspace{-2em}
\begin{subfigure}[t]{0.33\textwidth}
\begin{lstlisting}
c := 0; i := 0;
while (true) {
  fork i w();
  if (i > 0) {
    join i-1;
  }
  i := i + 1;
}
\end{lstlisting}
\vspace{-2mm}
\caption{The initial \texttt{main} thread}
\end{subfigure}\\

\begin{subfigure}[t]{0.33\textwidth}
\begin{lstlisting}
c := c + i;
assert c <= 2 * i;
c := c - i;
\end{lstlisting}
\vspace{-2mm}
\caption{The worker thread \texttt{w}}
\end{subfigure}
\vspace{-2mm}
    \caption{Program with dynamic thread management}
    \label{fig:example}
\end{wrapfigure}
The number of threads changes during execution and may depend
on input data. %, and cannot (easily) be determined statically.
As an example, consider the program in \cref{fig:example}.
In a loop, the \texttt{main} thread creates new worker threads and assigns them the thread ID \texttt{i},
i.e., the current iteration.
%The newly created threads run concurrently with all existing threads.
In all iterations but the first, the \texttt{main} thread also \emph{joins} the thread with ID \texttt{i-1},
i.e., the thread created in the previous iteration.
The \texttt{join} statement blocks until that thread terminates, and destroys the terminated thread.
%\\
%Many programming languages and multi-threading APIs allow concurrent programs to use \emph{dynamic thread management}.
Dynamic thread management in this style is widely used in practice.
%
%Data aspects, even input data, may influence how many threads are active (i.e., have been created but not yet destroyed) throughout an execution of the program.
For instance, the threads extension of the POSIX standard~\cite{pthreads} specifies such an API (commonly known as \emph{pthreads}),
as do Java~\cite{java:threads} and the .NET framework~\cite{dotnet:threads}.
%

%%%%%%%%%%%%%%%%%%%%%%%%%%%%%%%%%%%%%%%%%%%%%%%%%%%%%%%%%%%%%%%%%%%%%%%%%%%%%%%%
% Why is verification of programs with dynamic thread management challenging?
%%%%%%%%%%%%%%%%%%%%%%%%%%%%%%%%%%%%%%%%%%%%%%%%%%%%%%%%%%%%%%%%%%%%%%%%%%%%%%%%
%
Automated verification of programs with dynamic thread management is
%both crucial and
challenging.
%
% - undecidable how many threads created
%One can not easily
The control flow of the program (which code will be executed in what order) is not immediately apparent from the syntax.
In fact, it is already an undecidable problem
to determine statically \emph{how many} and \emph{which} threads will be created at runtime,
and \emph{when} or \emph{if} a thread will be joined.
%
% - schwierigkeiten den Kontrollfluss zu beschreiben
For instance, for the program in \cref{fig:example},
a control flow analysis must determine how many instances of the worker thread \texttt{w} can be active at the same time,
and consider all interleavings of their actions.
Furthermore, the analysis has to figure out that the thread joined in line~5 is always the thread created in the previous iteration,
hence the program's control flow does not include interleavings in which the joined thread executes another action after the \texttt{join} statement.
%
% - Problem weil unsere automatischen tools beschreibung von Kontrollfluss als input nehmen
This difficulty of determining the control flow presents a challenge for many \emph{software model checking} techniques,
which typically require a representation of a program's control flow as input.

% - Außerdem brauchen informationen über den Zustandsraum und der ist noch gar nicht bekannt.
Given a description of the control flow,
software model checkers reason about the program's data to construct an over-approximation of the reachable states.
This allows them to either prove the correctness of the program (if the over-approximation does not admit a specification violation)
or find bugs.
Here, we encounter a second challenge for programs with dynamic thread management:
Threads may have thread-local variables,
and program states must assign a value to $\texttt{x}_{t}$ for each thread $t$ and each local variable \texttt{x}.
Hence, if the number of threads created at runtime is unknown,
then so is the program's state space.
%For instance, if the worker threads in \cref{fig:example} had a local variable~\texttt{x},
%then the program states would have to assign a value to a variable $\texttt{x}_{t}$ for each worker thread $t$.
%
%Since threads may have thread-local data, even the space of program variables is unclear.
%Not knowing the program state space
This makes it challenging for software model checkers
to construct over-approximations of the reachable states that are precise enough to show correctness.

In this work, we present an approach for the automated verification of programs with dynamic thread management,
which overcomes these challenges.
Our focus is on programs where the number of thread that can be active at the same time is bounded.
We call the maximum number of threads (with the same thread template) that can be active at the same time the \emph{thread width} of the program.
%(Note that we do not bound the overall number of threads forked and possibly joined during an execution.)
%
Programs with a bounded thread width represent a sweet spot for automated verification:
Firstly, programs with a bounded thread width are still tractable for software model checking.
%
%without resorting to verification approaches for parametrized programs (i.e., programs where an unbounded number of threads is active at the same time).
%Verification of parametrized programs often requires coarse abstractions in order to finitely represent over-approximations of the program's state space.
%This results in \emph{incomplete} techniques, i.e., where many correct programs cannot be proven,
%and where a failure to verify a program does not necessarily imply that a bug can be found.
%
Secondly, the number of active threads is often small in practice
because programs are most efficient if the number of active threads is at most the number of CPU cores.
% because in many programs the number of active thread is related to the
% due to the physical constraints of target machines and the runtime overhead of thread creation.
This is in contrast to many other discrete characteristics of executions,
which can grow unboundedly depending on input data (e.g. length of executions, and consequently, the overall number of threads)
or are not under the programmer's control (e.g. number of context switches).

In order to apply software model checking techniques to the verification of programs with dynamic thread management,
and overcome the challenges laid out above,
we reduce the verification problem for a program with dynamic thread management
to a series of verification problems with a fixed number of threads.
The reduction is \emph{sound}:
If our approach concludes that a program with dynamic thread management is correct, then \emph{all} executions of the program are indeed correct.
To ensure soundness for all programs,
our approach attempts to identify the given program's thread width (by trying out different values)
and then \emph{verifies} that the program indeed has a certain thread width.
We again reduce this verification problem to a verification problem with a fixed number of threads.
Note that the thread width is not supplied by the user of the analysis.
%as would for instance be the case in bounded model checking.
Instead, the thread width is a property of the input program, derived from the program's semantics.
%This allows us to guarantee soundness of our approach for all input programs (with or without bounded thread width).
%
%Moreover,
Our approach is \emph{complete} for programs with a bounded thread width, and for incorrect programs.
The key technical step in our approach is
%In order to reduce the correctness of a program with dynamic thread management to verification problems for programs with a fixed number of threads,
%we present
\emph{petrification}, a transformation from programs with dynamic thread management to \emph{Petri programs}~\cite{vmcai2021} with a fixed number of threads.
Petrification separates the control flow of the program (including thread management) from data aspects,
and imposes a given \emph{thread~limit}~$\beta$ on the number of active threads.
The control flow of a Petri program is encoded as a Petri net,
while the manipulation of data (i.e., variable values) is expressed by labeling transitions of the Petri net with assignments and guards.
The resulting Petri program faithfully represents executions of the original program where at most $\beta$ threads are active at the same time.
%Let us stress that the number of threads created and destroyed over the whole execution is unbounded, as is the length of the execution.

%
%This separation of thread management (and other control flow aspects) from data enables us to apply existing verification algorithms for programs with a fixed number of threads
%to the \emph{petrified} program.
%The limit on the number of active threads is a \emph{parameter} of petrification.

We define two separate specifications for the petrified program:
a \emph{safety specification} which expresses correctness of all executions encoded by the Petri program,
and a \emph{bound specification}, expressing that
%the program never exceeds the thread limit $\beta$ (i.e.,
the program's thread width is bounded by $\beta$.
Both specifications can be verified independently, using an existing algorithm for the verification of Petri programs~\cite{vmcai2021}.
We thereby reduce the existence of a bound on the program's thread width to a safety specification for Petri programs;
checking this specification does not require any changes to the underlying verification algorithm.
We present a verification algorithm for programs with dynamic thread management,
that repeatedly petrifies a given program (with different thread limits $\beta$)
and invokes a Petri program verification algorithm~\cite{vmcai2021} to check the safety and bound specifications.
%To verify a given program with dynamic thread management,
%
%To verify a given program with dynamic thread management,
%we present a verification algorithm that makes use of petrification.
%The algorithm proves correctness by repeatedly petrifying the given program (with different thread limits $\beta$)
%and invoking a Petri program verification algorithm~\cite{vmcai2021} to check the safety and bound specifications.
%When it has shown that some thread limit $\beta$ is indeed an upper bound for the thread width,
%the algorithm concludes correctness of the input program.
%Let us stress again that this bound on the thread width is not chosen by the user of the analysis,
%but rather formally proven, based on the semantics of the input program.
We investigate several variants of this verification algorithm and their differences.% theoretically.
%and empirically evaluate the feasibility of the overall approach as well as the practical differences between the variants.

Our approach is implemented in the program analysis and verification framework \textsc{Ultimate}.
The implementation verifies concurrent C programs using the POSIX threads (\emph{pthreads}) standard for multi-threading.
Our approach is practicable and can be successfully used for automated verification.

\tikzstyle{transition}=[rectangle,thick,fill=black,minimum height=6mm,minimum width=3.5pt,inner xsep=0]

\begin{figure}
	\begin{tikzpicture}[thick,scale=1,node distance=0.5cm,
				trans/.style={->,>=angle 45,thick},
                place/.style={draw,circle,inner sep=1.5mm}
				]

        % Main Thread
        \node[place,label=below:$\ell_0$] (l0) {};
        \node[token] (l0token) at (l0) {};
        \node[transition,right=0.5cm of l0,label=above:$\stsmcol{c:=0; i:=0}{thinit}$] (t1) {};
        \node[place,right=0.5cm of t1,label=below:$\ell_{1}$] (l1) {};

        % NotInUse
        \node[place,label=above:$\notinuse{w}{1}$,above=3cm of l0,xshift=1.8cm] (n1) {};
        \node[token] (n1token) at (n1) {};
        \node[place,label=below:$\notinuse{w}{2}$,below=3cm of l0,xshift=1.8cm] (n2) {};
        \node[token] (n2token) at (n2) {};

        % Forks
        \node[transition,right=1.5cm of n1,yshift=-0.6cm,label={[left, xshift=-0.2cm, yshift=-0.4cm]$\stsmcol{\idvar{w}{1}:=i}{thinit}$}] (f1) {};
        \node[transition,right=1.5cm of n2,yshift=0.6cm,label={[left, xshift=-0.2cm, yshift=-0.2cm]$\stsmcol{\idvar{w}{2}:=i}{thinit}$}] (f2) {};
        \node[transition,right=0.9cm of l1,label=above:$\stsmcol{true}{thinit}$] (fe) {};

        % InUse / Insuff
        \node[place,label=above:$\inuse{w}{1}$,right=2cm of f1] (i1) {};
        \node[place,label=below:$\inuse{w}{2}$,right=2cm of f2] (i2) {};
        \node[place,right=0.9cm of fe,label=below:$\insuff{w}$, accepting] (insuff) {};

        \node[place,right=1.4cm of insuff,label=below:$\ell_{2}$] (l2) {};
        \node[transition,right=0.4cm of l2,label=below:$\stsmcol{i>0}{thinit}$] (if) {};
        \node[transition,above=0.5cm of if,label=above:$\stsmcol{i<=0}{thinit}$] (else) {};
        \node[place,right=0.4cm of if,label=right:$\ell_{3}$] (l3) {};
        \node[place,right=0.9cm of l3,label=right:$\ell_{4}$] (l4) {};
        \node[transition,below=1cm of if,label=below:$\stsmcol{i:=i+1}{thinit}$] (inc) {};

        % Joins
        \node[transition,right=7.6cm of n1,label=right:$\stsmcol{\idvar{w}{1}==i-1}{thinit}$] (j1) {};
        \node[transition,right=7.6cm of n2,label=right:$\stsmcol{\idvar{w}{2}==i-1}{thinit}$] (j2) {};

        % Thread 1
        \node[place,label=above:$\ell_{0,1}$,xshift=0.4cm,above=2.5cm of f1] (w01) {};
        \node[transition,right=0.5cm of w01,label=below:$\stsmcol{c:=c+i}{th1}$] (t01) {};
        \node[place,label=above:$\ell_{1,1}$,right=0.5cm of t01] (w11) {};
        \node[transition,right=0.5cm of w11,label=below:$\stsmcol{c<=2*i}{th1}$] (t11) {};
        \node[transition,above=0.5cm of t11,label=above:$\stsmcol{c>2*i}{th1}$] (te1) {};
        \node[place,label=right:$\lightning_1$,right=0.5cm of te1,accepting] (we1) {};
        \node[place,label=above:$\ell_{2,1}$,right=0.5cm of t11] (w21) {};
        \node[transition,right=0.5cm of w21,label=below:$\stsmcol{c:=c-i}{th1}$] (t21) {};
        \node[place,label=above:$\ell_{3,1}$,right=0.5cm of t21] (w31) {};
        \draw [trans] (f1) to (w01);
        \draw [trans] (w01) to (t01);
        \draw [trans] (t01) to (w11);
        \draw [trans] (w11) to (t11);
        \draw [trans] (w11) to (te1);
        \draw [trans] (te1) to (we1);
        \draw [trans] (t11) to (w21);
        \draw [trans] (w21) to (t21);
        \draw [trans] (t21) to (w31);

        % Thread 2
        \node[place,label=below:$\ell_{0,2}$,xshift=0.5cm,below=2.5cm of f2] (w02) {};
        \node[transition,right=0.5cm of w02,label=above:$\stsmcol{c:=c+i}{th2}$] (t02) {};
        \node[place,label=below:$\ell_{1,2}$,right=0.5cm of t02] (w12) {};
        \node[transition,right=0.5cm of w12,label=above:$\stsmcol{c<=2*i}{th2}$] (t12) {};
        \node[transition,below=0.5cm of t12,label=below:$\stsmcol{c>2*i}{th2}$] (te2) {};
        \node[place,label=right:$\lightning_2$,right=0.5cm of te2,accepting] (we2) {};
        \node[place,label=below:$\ell_{2,1}$,right=0.5cm of t12] (w22) {};
        \node[transition,right=0.5cm of w22,label=above:$\stsmcol{c:=c-i}{th2}$] (t22) {};
        \node[place,label=below:$\ell_{3,2}$,right=0.5cm of t22] (w32) {};
        \draw [trans] (f2) to (w02);
        \draw [trans] (w02) to (t02);
        \draw [trans] (t02) to (w12);
        \draw [trans] (w12) to (t12);
        \draw [trans] (w12) to (te2);
        \draw [trans] (te2) to (we2);
        \draw [trans] (t12) to (w22);
        \draw [trans] (w22) to (t22);
        \draw [trans] (t22) to (w32);
        \draw [trans] (w31) to (j1);

		\draw [trans] (l0) to (t1);
		\draw [trans] (t1) to (l1);
        \draw [trans] (n1) to (f1);
		\draw [trans] (l1) to (f1);
        \draw [trans] (l1) to (f2);
        \draw [trans] (i1) to (f2);
        \draw [trans] (n2) to (f2);
        \draw [trans] (l1) to (fe);
        \draw [trans] (i1) to (fe);
        \draw [trans] (i2) to (fe);
        \draw [trans] (fe) to (insuff);
        \draw [trans] (f1) to (i1);
        \draw [trans] (f1) to (l2);
        \draw [trans] (f2) to (i1);
        \draw [trans] (f2) to (i2);
        \draw [trans] (f2) to (l2);
        \draw [trans] (l2) to (if);
        \draw [trans, bend left=40] (l2) to (else);
        \draw [trans] (if) to (l3);
        \draw [trans, bend left=25] (else) to (l4);
        \draw [trans,bend left=30] (l4) to (inc);
        \draw [trans,bend left=17] (inc) to (l1);
        \draw [trans] (i1) to (j1);
        \draw [trans] (l3) to (j1);
        \draw [trans,bend right=16] (j1) to (n1);
        \draw [trans] (j1) to (l4);
        \draw [trans] (l3) to (j2);
        \draw [trans] (i2) to (j2);
        \draw [trans,bend left=16] (j2) to (n2);
        \draw [trans] (j2) to (l4);
        \draw [trans] (w32) to (j2);

	\end{tikzpicture}
    \caption{Petri net for the example program from \cref{fig:example}, with two instances for the worker thread.}
    \label{fig:petri-net}
\end{figure}
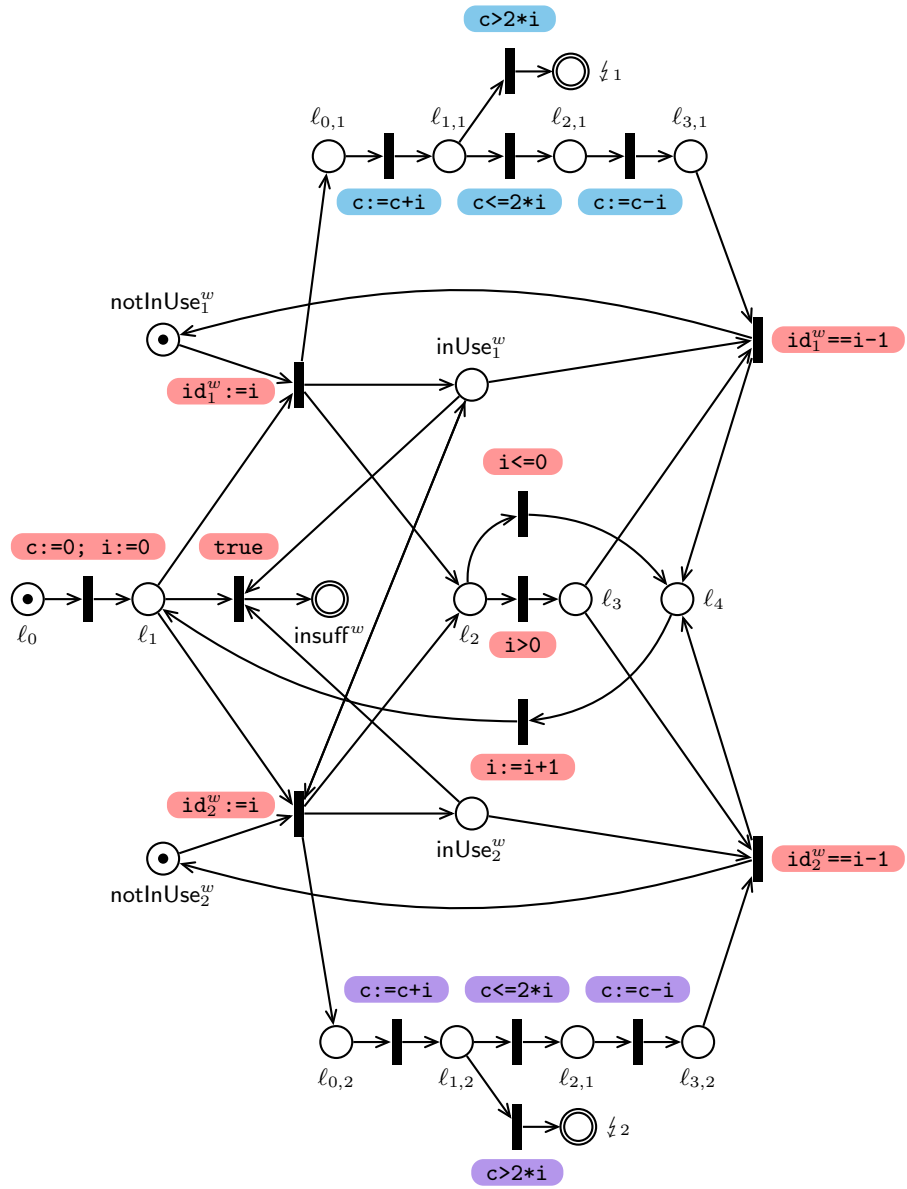

\vspace*{-5mm}
\subsubsection{Example.}
As an illustration of our approach, consider again the program in \cref{fig:example}.
%At the start of the program, only the \texttt{main} thread is running.
%In a loop, the \texttt{main} thread creates new worker threads and assigns them the thread ID \texttt{i},
%i.e., the current iteration.
%The newly created thread runs concurrently with all existing threads.
%In all iterations but the first, the \texttt{main} thread subsequently \emph{joins} the thread with ID \texttt{i-1},
%i.e., the thread created in the previous iteration.
%The \texttt{join} statement blocks until that thread has terminated, destroys the terminated thread, and continues.
The worker threads each perform the same computation involving the global variables \texttt{c} and \texttt{i},
and assert that $\texttt{c}\leq 2\cdot\texttt{i}$.
The verification problem consists of showing that this condition always holds when any worker thread executes the \texttt{assert} statement.

As the first verification step, we \emph{petrify} the program.
If we chose a thread limit $\beta=2$ -- i.e., at most two worker threads can run at the same time --,
petrification produces the Petri program in \cref{fig:petri-net}.
We see the main thread (in the middle, in red), as well as two instances of the worker thread (at the top resp.\ at the bottom, in blue resp.\ purple).
Transitions (black vertical bars) encode the possible control flow, and are annotated with assignments and guards.
The thread ID is tracked through specially-introduced variables \idvar{w}{1} and \idvar{w}{2}.
The transitions labeled \stsmcol{\texttt{\idvar{w}{k}:=i}}{thinit} correspond to the \texttt{fork} statement,
whereas the transitions labeled \stsmcol{\texttt{\idvar{w}{k}==i-1}}{thinit} correspond to the \texttt{join} statement.

We employ the Petri program verification algorithm~\cite{vmcai2021}
to show that the petrified program satisfies both its bound and safety specification.
Intuitively, this means that no run (or firing sequence) of the Petri net,
that also obeys the assignments and guards of the transitions (in a sense made formal below),
can reach the place $\insuff{w}$ (for the bound specification) resp.\ the place $\lightning_1$ or $\lightning_2$ (for the safety specification).
Intuitively, the former holds because at the start of any iteration, only the worker thread created in the previous iteration is active.
The latter then follows, because all worker threads created in even earlier iterations must have terminated,
and their overall effect was to either leave \texttt{c} unchanged or decrease it (if \texttt{i} was incremented in between the two assignments to \texttt{c}).
If $n\in\{0,1,2\}$ active worker threads have executed the first but not the second assignment,
we have $\texttt{c} \leq n \cdot \texttt{i}$.
% the value of \texttt{c} is at most $n$ times that of \texttt{i}.
%
We conclude that the program in \cref{fig:example} is correct.

\vspace*{-5mm}
\subsubsection{Related Work.}
%%%%%%%%%%%%%%%%%%%%%%%%%%%%%%%%%%%%%%%%%%%%%%%%%%%%%%%%%%%%%%%%%%%%%%%%%%%%%%%%
% Previous approaches: thread-modular reasoning vs bounded verification
%%%%%%%%%%%%%%%%%%%%%%%%%%%%%%%%%%%%%%%%%%%%%%%%%%%%%%%%%%%%%%%%%%%%%%%%%%%%%%%%
%
%Previous approaches to the verification of programs with dynamic thread management, or more generally programs with an unbounded number of threads,
%can roughly be divided into two categories:
%\emph{thread-modular verification} and \emph{bounded verification}.

One approach for the analysis of concurrent programs found in the literature is \emph{bounded verification}.
%
%
%%%%%%%%%%%%%%%%%%%%%%%%%%%%%%%%%%%%%%%%%%%%%%%%%%%%%%%%%%%%%%%%%%%%%%%%%%%%%%%%
% Bounded verification
%%%%%%%%%%%%%%%%%%%%%%%%%%%%%%%%%%%%%%%%%%%%%%%%%%%%%%%%%%%%%%%%%%%%%%%%%%%%%%%%
In bounded verification, the user of an analysis supplies some bound,
and the analysis then considers only those behaviours of the program below the given bound.
%of a program imposes limits on the behaviours of the program.
For instance, one may only analyze executions up to a certain length~\cite{dartagnan:cav,deagle,dartagnan:sv-comp}, with at most a given number of threads~\cite{cseq:unbounded-switch}, or a with a limited number of context switches~\cite{cseq:cav,kiss}.
In each case, the analysis only covers a fragment of the program behaviours.
For programs that exceed the imposed limit, such analyses are often successful in finding bugs, but cannot soundly conclude that the program is correct.
By contrast, we present an approach
that uses a form of bounded verification as a sub-procedure (the bound is given by the thread limit $\beta$),
and yet gives full correctness guarantees.

%%%%%%%%%%%%%%%%%%%%%%%%%%%%%%%%%%%%%%%%%%%%%%%%%%%%%%%%%%%%%%%%%%%%%%%%%%%%%%%%
% Cutoff theorems
%%%%%%%%%%%%%%%%%%%%%%%%%%%%%%%%%%%%%%%%%%%%%%%%%%%%%%%%%%%%%%%%%%%%%%%%%%%%%%%%

%We have already discussed analyses that impose artificial bounds chosen by the user of an analysis.
%Such techniques can find bugs, but cannot soundly show correctness.

Our approach shares some similarities with works that combine bounded verification
%In order to fully guarantee correctness, bounded verification can be combined
with the search for a \emph{cutoff point}~\cite{DBLP:conf/podc/ClarkeG87,DBLP:conf/concur/ClarkeTTV04,DBLP:conf/cade/EmersonK00,DBLP:conf/icse/YangL10}:
A finite bound such that correctness of the fragment of program behaviours up to the cutoff point
also implies correctness of all other program behaviours.
\todo{explain more?}
In contrast to such works, the satisfaction of a certain thread width bound by a program is wholly independent of the verified property, the logic in which it is expressed, or how it is proven.
However, given a suitable cutoff point for the thread width bound,
one may employ petrification to produce and subsequently verify the corresponding bounded instance.

%%%%%%%%%%%%%%%%%%%%%%%%%%%%%%%%%%%%%%%%%%%%%%%%%%%%%%%%%%%%%%%%%%%%%%%%%%%%%%%%
% Thread-modular verification
%%%%%%%%%%%%%%%%%%%%%%%%%%%%%%%%%%%%%%%%%%%%%%%%%%%%%%%%%%%%%%%%%%%%%%%%%%%%%%%%
There are many works concerned with the verification of \emph{parametrized programs},
i.e., program with an unbounded number of active threads.
For instance, in \emph{thread-modular verification}, the idea is
%One approach to verify programs with an unknown, and possibly unbounded number of threads is
to find a modular proof that allows for generalization to any number of threads.
For instance, thread-modular proofs at level $k$~\cite{jochen:thread-modular} generalize the non-interference condition of Owicki and Gries~\cite{owicki-gries}.
A thread-modular proof for a program with $k$ threads establishes correctness of the program for any number of threads.
Other forms of compositional verification, such as rely-guarantee reasoning, follow a similar approach.
A challenge in this setting is that such proof methods are often incomplete, and
while there is some work in automating the search for such a proof (e.g. the thread-modular abstract interpreter \textsc{Goblint}~\cite{goblint:sv-comp}),
software model checking based on compositional proofs is not yet as mature. %\todo{as what?}
%\todo{Cite \cite{popl/GuptaPR11,cav/GuptaPR11}?}
%
An additional challenge is that many such works are based on the
 setting of parametric programs, where an arbitrary number (not controlled by the program) of threads  execute concurrently (i.e., all threads start at the same time).
Supporting dynamic thread management often requires further work~\cite{goblint:thread-ids}.

On a technical level, the work most closely related to our approach~\cite{cpachecker:threads}
describes the implementation of the software model checker \textsc{CpaChecker}~\cite{cpachecker}.
Upfront, this implementation creates a fixed number of procedure clones.
It then performs a specialized configurable program analysis (CPA),
which treats the cloned procedures as threads.
%However, the described approach is more ad-hoc, based on typical C programs, whereas
In contrast, we describe a modular approach that
separates the reduction to programs with a fixed number of threads from the analysis,
and allows for different verification algorithms to be applied.
Our approach is grounded in a formal semantics of a concurrent programming language,
which allows for theoretical analysis.
%Furthermore, the description intermingles the reduction to a program with a fixed number of threads with the subsequent configurable program analysis, whereas

\vspace*{-5mm}
\subsubsection{Contributions.}
To summarize, our contributions are as follows:
\begin{itemize}
\item We identify the class of programs with a \emph{bounded thread width} as a sweet spot for software model checking,
both relevant and tractable.
\item We present an approach for the automated verification of programs with dynamic thread management,
  by reduction to a series of verification problems with a fixed number of threads.
  The approach is sound, and it is complete for programs with bounded thread width as well as incorrect programs.
%  \begin{itemize}
  \item The key technical contribution behind this reduction is \emph{petrification},
    a construction that captures the control flow of programs with dynamic thread management (up to a given thread limit $\beta$)
    as a Petri program~\cite{vmcai2021} with a fixed number of threads.
  \item We implemented our approach in the program verification framework \textsc{Ultimate} and showed that it is successful in practice.
 % \end{itemize}
\end{itemize}

\vspace*{-5mm}
\subsubsection{Roadmap.}
\Cref{sec:language} presents the syntax and semantics of \proglang, a simple language with dynamic thread management, and defines the verification problem.
In \cref{sec:petrification}, we quickly recap the essential notions of Petri programs and then introduce petrification, our transformation from \proglang to Petri programs.
We present the overall verification algorithms based on petrification in \cref{sec:verification}.
\Cref{sec:implementation} discusses the implementation of our approach in a verification tool for C programs,
and evaluates the practical feasibility.
We conclude in \cref{sec:conclusion}.

\section{A language for dynamic thread creation}
\label{sec:language}

In this section, we introduce \proglang, a simple imperative language for concurrent programs.
\proglang captures the essence of dynamic thread management.
We define syntax and semantics of \proglang, as well as the corresponding verification problem.
%In the subsequent sections, we present an approach to the verification of \proglang programs.

\vspace*{-5mm}
\subsubsection{Syntax.}
\proglang programs consist of a number of \emph{thread templates},
which have a unique name, and a body specifying the code to be executed by a thread.
We denote the finite set of all possible thread template names by $\threadTemplates$.
Further, let $\vars$ be a set of program variable names.
We assume as given a language of expressions over variables in $\vars$, e.g., the expressions defined by the \textsc{Smt-Lib} standard~\cite{smtlib}.
Let $x$ range over variables,
$\vartheta$ range over thread templates,
$e_\mathsf{int}$ range over integer-valued expressions,
$e_\mathsf{bool}$ range over boolean-valued expressions,
and $e$ range over both integer- and boolean-valued expressions.
The syntax of \proglang commands is defined by the following grammar:
\goodbreak
\[
\begin{array}{rcl}
    %P & ::= & C \Omega\\
    C & ::= &  \texttt{$x$:=$e$}%\\
    ~|~  \texttt{assume}\; e_\mathsf{bool}
    ~|~  \texttt{assert}\; e_\mathsf{bool}\\
    & |& \texttt{if\;($e_\mathsf{bool}$)\;\{\;$C$\;\}\;else\;\{\;$C$\;\}} ~|~
    \texttt{while\;($e_\mathsf{bool}$)\;\{\;$C$\;\}}\\
    & |& \texttt{fork} \; e_\mathsf{int}\; \vartheta() ~|~
    \texttt{join} \; e_\mathsf{int}\\
    & |&  C;C
  \end{array}
\]
The set of all commands is denoted by $\Command$.
The set of \emph{atomic statements}, $\atomicstmt$, is the set of all assignments and \texttt{assume} commands.
We call atomic statements as well as \texttt{fork} and \texttt{join} commands \emph{simple statements}.

The key feature of \proglang are the \texttt{fork} and \texttt{join} commands for dynamic thread management.
The statement \texttt{fork $e$ $\vartheta$()} creates a new thread,
whose code is given by the thread template $\vartheta$.
The expression $e$ is evaluated, and its current value is used as the \emph{thread ID} of the newly created thread.
All forked threads run concurrently, and their computation steps can be arbitrarily interleaved (i.e., we consider a sequential consistency semantics).
When a thread executes \texttt{join $e$}, the execution blocks until some other thread, whose thread ID equals the value of $e$, has terminated.
The terminated thread is destroyed.

This style of dynamic thread management is inspired by the POSIX threads (or \emph{pthreads}) API~\cite{pthreads},
in which a call to \texttt{pthread\_create} creates a new thread, with the thread template given by a function pointer.
In contrast to our \texttt{fork} command, \texttt{pthread\_create} also computes and returns a thread ID.
In \proglang programs, the computation of a (unique) thread ID can be implemented separately, and the resulting thread ID can be used by a \texttt{fork} command.
However, \proglang also allows multiple threads to share the same thread ID.
Generally, there is a tradeoff regarding uniqueness of thread IDs:
One can enforce uniqueness of thread IDs on the semantic level, and possibly design a type system that ensures this uniqueness.
This complicates the definition of the language, and slightly limits the expressiveness of the language.
On the other hand, allowing thread IDs to be non-unique leads to a simpler and more expressive language.
%,but complicates analyses of programs.
Though it slightly complicates some proofs in this paper, we choose the latter path.

%Let us quickly go over the other, not concurrency-related commands.
The semantics of the other, not concurrency-related commands, is standard:
An \texttt{assume} commands blocks execution if the given expression evaluates to $\false$, otherwise it has no effect.
An \texttt{assert} command \emph{fails} if the given expression evaluates to $\false$.
The intuition behind the assignment, \texttt{if-then-else} and \texttt{while} commands,
as well as the sequential composition (\texttt{;}) of commands is as expected.

We use the special values $\Omega$ to represent a command that has successfully terminated,
and $\lightning$ for a command that has failed (due to a violated \texttt{assert}).
For convenience, we extend the sequential composition by setting $C; \Omega := C$.
%Intuitively, if the first command in a sequential composition has successfully terminated, it remains to execute the second command.
%If the first command in a sequential composition fails, the sequential composition fails.

A program is given by a tuple $\prog = (\body{}, \main, \globalVars)$,
consisting of a mapping $\body{}$ that associates each thread template name $\vartheta$ with a command $\body{\vartheta}$,
a thread template $\main \in \threadTemplates$ for the main thread,
and a set of global variables $\globalVars$.
For each thread template $\vartheta\in\threadTemplates$,
the command $\body{\vartheta}$ identifies the code executed by instances of this thread template.
This command may refer to the variables in $\globalVars$,
as well as to any other variables (which are implicitly assumed to be thread-local).
We do not require variables to be declared, and uninitialized variables can have arbitrary values.

\vspace*{-5mm}
\subsubsection{Semantics.}
We define the (small-step) semantics of \proglang, in the style of structural operational semantics.
To this end, we first introduce the notion of \emph{local} and \emph{global configurations} of a given \proglang program $\prog = (\body{},\main,\globalVars)$.

%Configuration
A \emph{local configuration} is a quadruple $\langle X, \vartheta, t, s \rangle$ consisting of
some $X$ that is either a remainder program left to execute ($X \in \Command$), or a special value to indicate termination ($\Omega$) or failure ($\lightning$),
a thread template name $\vartheta$,
a thread ID $t\in\mathbb{Z} \cup \{\bot\}$,
and a local state $s : \vars \setminus \globalVars \to \mathbb{Z} \cup \{\true,\false\}$.
The thread ID $\bot$ is used exclusively for the start thread.
In general, thread IDs need not be unique.
We keep track of the thread ID in the local configuration,
in order to determine which threads can be joined when another thread executes a \texttt{join} command.
If multiple threads with the same ID are ready to be joined, one thread is chosen nondeterministically.

A \emph{global configuration} is a pair $(M, g)$ of a multiset $M$ of local configurations, and a global state $g : \globalVars \to \mathbb{Z} \cup \{\true,\false\}$.
We use a multiset to reflect the fact that several running threads could have the same local configuration.
A global configuration $(M, g)$ is \emph{initial}
if $M = \lbag \langle \body{\main}, \main, \bot, s \rangle \rbag$
for any local state $s$;
the global state $g$ is also arbitrary.
(The symbols $\lbag \ldots \rbag$ denote a multiset containing the listed elements.)
%
%$post(S, i, C)$ defined for set of configurations $S$, thread id $i$ and atomic command $C$

%\subsection{Transitions}
\Cref{fig:sem-def} defines the small-step structural operational semantics of our language as a transition relation $(M, g) \semtrans{\st} (M', g')$ over global configurations $(M,g), (M',g')$ and simple statements $\st$.
We assume here that, given a mapping $\tilde{s} : \vars \to \mathbb{Z}\cup\{\true,\false\}$,
we can evaluate an expression $e$ to some value $\sem{e}^{\tilde{s}} \in \mathbb{Z}\cup\{\true,\false\}$.
In particular, if $s$ is a local state and $g$ is a global state, we can set $\tilde{s} := s \cup g$.
\begin{figure}
\begin{gather*}
  % Assume statements
  \infer[\srule{Assume}]{
    \lbag \langle \texttt{assume $e$};X,\vartheta,t,s\rangle \rbag,g \semtrans{\texttt{assume $e$}} \lbag \langle X,\vartheta,t,s\rangle \rbag,g
  }{
    \sem{e}^{s\cup g} = \true
  }
  \\[1em]
%
  % Assignments
  \infer[\srule{AssignGlobal}]{
    \lbag \langle \texttt{$x$:=$e$};X,\vartheta,t,s\rangle \rbag,g \semtrans{\texttt{$x$:=$e$}} \lbag \langle X,\vartheta,t,s\rangle\rbag,g[x\mapsto \sem{e}^{s\cup g}]
  }{
    x \in \globalVars
  }
  \\[1em]
  \infer[\srule{AssignLocal}]{
    \lbag\langle \texttt{$x$:=$e$}; X,\vartheta,t,s\rangle\rbag,g \semtrans{\texttt{$x$:=$e$}} \lbag\langle X,\vartheta,t,s[x\mapsto \sem{e}^{s\cup g}]\rangle\rbag,g
  }{
    x \notin \globalVars
  }
  \\[1em]
%
  % Assert statements
  \infer[\srule{Assert1}]{
    \lbag \langle \texttt{assert $e$}; X,\vartheta,t,s\rangle \rbag,g \semtrans{\texttt{assume $e$}} \lbag \langle X,\vartheta,t,s\rangle \rbag, s
  }{
    \sem{e}^{s\cup g} = \true
  }
  \\[1em]
  \infer[\srule{Assert2}]{
    \lbag \langle \texttt{assert $e$}; X,\vartheta,t,s\rangle \rbag,g \semtrans{\texttt{assume !$e$}} \lbag \langle \lightning,\vartheta,t,s\rangle \rbag, g
  }{
    \sem{e}^{s\cup g} = \false
  }
  \\[1em]
%
%  % Sequential Composition
%  \infer[\srule{Seq}]{
%    \lbag \langle C_1;C_2,\vartheta,t,s\rangle \rbag \uplus M, g \semtrans{\st} \lbag \langle X;C_2,\vartheta,t,s'\rangle \rbag \uplus M', g'
%  }{
%    \lbag \langle C_1,\vartheta,t,s\rangle\rbag \uplus M, g \semtrans{\st} \lbag \langle X,\vartheta,t,s'\rangle\rbag \uplus M', g'
%  }
%  \\[1em]
%
%  \infer[\srule{Seq1}]{
%    \{\langle t,C_1;C_2,s_l\rangle\} \uplus M, s_g \stackrel{\st}{\longrightarrow} \{\langle t,C_1';C_2,s_l'\rangle\} \uplus M', s_g'
%  }{
%    \{\langle t,C_1,s_l\rangle\} \uplus M, s_g \stackrel{\st}{\longrightarrow} \{\langle t,C_1',s_l'\rangle\} \uplus M', s_g'
%  }
%  \\[1em]
%%
%  \infer[\srule{Seq2}]{
%    \{\langle t,C_1;C_2,s_l\rangle\} \uplus M, s_g \stackrel{\st}{\longrightarrow} \{\langle t,C_2,s_l'\rangle\} \uplus M', s_g'
%  }{
%    \{\langle t,C_1,s_l\rangle\} \uplus M, s_g \stackrel{\st}{\longrightarrow} \{\langle t,\Omega,s_l'\rangle\} \uplus M', s_g'
%  }
%  \\[1em]
%%
%  \infer[\srule{Seq3}]{
%    \{\langle t,C_1;C_2,s_l\rangle\} \uplus M, s_g \stackrel{\st}{\longrightarrow} \{\langle t,\lightning,s_l'\rangle\} \uplus M', s_g'
%  }{
%    \{\langle t,C_1,s_l\rangle\} \uplus M, s_g \stackrel{\st}{\longrightarrow} \{\langle t,\lightning,s_l'\rangle\} \uplus M', s_g'
%  }
%  \\[1em]
%
  % If-Then-Else
  \infer[\srule{Ite1}]{
    \lbag\langle \texttt{if\,($e$)\,\{\,$C_1$\,\}\,else\,\{\,$C_2$\,\}}; X, \vartheta, t,s\rangle\rbag, g \semtrans{\texttt{assume $e$}} \lbag \langle C_1; X,\vartheta,t,s\rangle\rbag,g
  }{
    \sem{e}^{s\cup g} = \true
  }
  \\[1em]
  \infer[\srule{Ite2}]{
    \lbag \langle \texttt{if\,($e$)\,\{\,$C_1$\,\}\,else\,\{\,$C_2$\,\}}; X, \vartheta, t,s\rangle\rbag, g \semtrans{\texttt{assume !$e$}} \lbag\langle C_2; X, \vartheta, t,s\rangle\rbag,g
  }{
    \sem{e}^{s\cup g} = \false
  }
  \\[1em]
%
  % While statements
  \infer[\srule{While1}]{
    \lbag\langle \texttt{while\,($e$)\,\{\,$C$\,\}}; X, \vartheta, t,s\rangle\rbag, g \semtrans{\texttt{assume $e$}} \lbag\langle \texttt{$C\,;$\,while\,($e$)\,\{\,$C$\,\}}; X,\vartheta, t,s\rangle\rbag,g
  }{
    \sem{e}^{s\cup g} = \mathit{true}
  }
  \\[1em]
  \infer[\srule{While2}]{
    \lbag\langle \texttt{while\,($e$)\,\{\,$C$\,\}}; X, \vartheta, t,s\rangle\rbag, g \semtrans{\texttt{assume !$e$}} \lbag \langle X, \vartheta, t,s\rangle\rbag,g
  }{
    \sem{e}^{s\cup g} = \mathit{false}
  }
  \\[1em]
%
  % Fork statements
  \infer[\srule{Fork}]{
    \lbag \langle \texttt{fork\;$e$\;$\vartheta'$()}; X, \vartheta, t, s\rangle \rbag, g \semtrans{\texttt{fork\;$e$\;$\vartheta'$()}} \lbag \langle X, \vartheta, t,s\rangle,\langle \body{\vartheta'},\vartheta', \sem{e}^{s\cup g}, s' \rangle \rbag, g
  }{}
  \\[1em]
%
  % Join Statements
  \infer[\srule{Join}]{
    \lbag \langle \texttt{join\;$e$}; X, \vartheta, t, s\rangle, \langle\Omega, \vartheta', \sem{e}^{s\cup g}, s'\rangle \rbag, g \semtrans{\texttt{join\;$e$}} \lbag \langle X,\vartheta,t,s\rangle\rbag , g
  }{
  }
  \\[1em]
%
  % Frame rule
  \infer[\srule{Frame}]{
    M_1\uplus M_2,g \stackrel{\st}{\longrightarrow} M_1'\uplus M_2,g'
  }{
    M_1,g \stackrel{\st}{\longrightarrow} M_1',g'
  }
\end{gather*}
\caption{The definition of the small-step semantic transition relation. Assume that $C,C_1,C_2\in\Command$, $X \in \Command \cup \{\Omega\}$, $\vartheta,\vartheta'\in\threadTemplates, t\in\mathbb{Z}\cup\{\bot\}$, $s,s'$ are local states, $g,g'$ are global states, and $e$ is an expression.}
\label{fig:sem-def}
\end{figure}
Given the semantic transition relation, we define:
\begin{definition}[Execution]
  An \emph{execution} is a sequence of global configurations and statements
%	Interleaved sequence of pairs of multiset and global state and atomic statements
$
  (M_0, g_0) \semtrans{\st_1} \ldots \semtrans{\st_n} (M_n, g_n)
$
where $(M_0,g_0)$ is initial.
\end{definition}

%The correctness notion for \proglang programs follows directly:

\begin{definition}[Correctness]
  An execution $(M_0, g_0) \semtrans{\st_1} \ldots \semtrans{\st_n} (M_n, g_n)$ is \emph{erroneous} if $\lightning$ occurs in any local configuration of any $M_i$.
  A \proglang program $\prog$ is \emph{correct}
  if there does not exist any erroneous execution of $\prog$.
\end{definition}

%We conclude the section by introducing another notion that will play a crucial role in our verification approach:

\begin{definition}[Thread Width]
We say that the \emph{thread width of the execution} $(M_0, g_0) \semtrans{\st_1} \ldots \semtrans{\st_n} (M_n, g_n)$
is the maximum number $\beta\in\mathbb{N}$
such that some $M_i$ contains $\beta$ local configurations with the same thread template $\vartheta$.
The \emph{thread width of the program $\prog$} is the supremum over the thread widths of all executions of $\prog$.
%
%
%\begin{itemize}\
% \item We say that \emph{$k$ thread instances of procedure $p$ are sufficient for $E$ in step $i$} if $M_i$ contains less than or equal to $k$ elements of the form $\langle t, C, l\rangle$
%such that $C$ is a statement of procedure $p$.
% \item We say that \emph{$k$ thread instances of procedure $p$ are sufficient for $E$} if for all steps $1\leq n$ if it holds that $k$ thread instances of procedure $p$ are sufficient for $E$ in step $i$
% \item We say that \emph{$k$ thread instances of procedure $p$ are sufficient for a concurrent program $P$} if for all execution $E$ of $P$ it holds that $k$ thread instances of procedure $p$ are sufficient for $E$.
% \item We say that \emph{$k$ thread instances are sufficient for a concurrent program $P$} if it holds for all procedures $p$ that $k$ thread instances of procedure $p$ are sufficient for a concurrent program $P$.
%\end{itemize}
\end{definition}
The thread width plays a crucial role in our verification approach.
Note that the thread width of an execution is always a natural number.
However, the thread width of a program might be infinite.
%If the thread width of a program $\prog$ is finite,
%We call such a program \emph{unbounded}.

\begin{example}
  The thread width of the program in \cref{fig:example} is $\beta=2$.
  Although the program might unboundedly often execute a \texttt{fork} command,
  there are always at most two worker threads active at the same time.
\end{example}

%\begin{lemma}
%If $\petriProg{n}{\prog}$ is safe iff $n$  thread instances are sufficient for a concurrent program $\prog$.
%\end{lemma}

% $$\frac{
% \langle t_1, C_1\rangle, \cdots,  \langle t_i,\texttt{fork iexp P()} C_i\rangle, \cdots, \langle t_1, C_1\rangle, \;\conf
% }{
% 	\langle t_1, C_1\rangle, \cdots,  \langle t_i,C_i\rangle, \langle \sem{expr}^{t_i}_c , C_P \rangle, \cdots, \langle t_1, C_1\rangle, \;\conf
% }$$
%
%
% $$\frac{
% \langle t_1, C_1\rangle, \cdots,  \langle t_i,\texttt{join iexp } C_i\rangle, \langle \sem{expr}^{t_i}_\conf , \Omega \rangle, \cdots, \langle t_1, C_1\rangle\;\conf
% }{
% 	\langle t_1, C_1\rangle, \cdots,  \langle t_i,C_i\rangle \cdots, \langle t_1, C_1\rangle\;\conf
% }$$
% !TEX root = ../main.tex

%\section{Petri programs}
\section{Petrification}
\label{sec:petrification}

In this section, we describe a process called \emph{petrification},
which transforms a $\proglang$ program $\prog$ into a representation suitable for verification algorithms.
Specifically, we build on the formalism of~\cite{vmcai2021},
and transform $\prog$ into a so-called \emph{Petri program}.
The resulting Petri program can be verified using the algorithm presented in~\cite{vmcai2021}.
%
%First, we introduce a formalism to represent the control flow of $\prog$ with a fixed limit on the number of threads executing concurrently at runtime.
The petrified program captures the control flow of $\prog$,
%This formalism
and maintains the concurrent nature of the program (as opposed to, say, an interleaving model).
%,without considering data flow aspects.
%
%Afterwards we present the construction for a Petri program that implements this control flow.
Petrification is parametrized in an upper limit on the number of threads.
In \cref{sec:verification}, we show how an iterative verification algorithm can manipulate this parameter
in order to determine a (semantic) upper bound on the number of active threads (i.e. thread width), if it exists.

\subsection{Petri Programs}
\label{sec:petri-programs}

Before presenting our construction, we give a brief recap on Petri programs.
\begin{definition}[Petri Programs]
  A Petri program is given by a 5-tuple $
	\N = (P, T, F, m_\mathsf{init}, \lab)
  $,
  where
  %$\Sigma$ is an alphabet with $\Sigma \subseteq \atomicstmt$,
  $P$ is a finite set of \emph{places},
  $T$ is a finite set of \emph{transitions} with $P\cap T=\emptyset$,
  $F\subseteq (P\times T)\cup(T\times P)$ is a \emph{flow relation},
  $m_\mathsf{init}: P\to \mathbb{N}$ is an \emph{initial marking},
  and $\lab:T\to \atomicstmt$ is a \emph{labeling} of transitions.
  %and $P_\mathsf{fin}\subseteq P$ is a set of \emph{accepting places}.
\end{definition}
Equipped with Petri net semantics,
a Petri program defines a set of \emph{traces} (i.e., sequences of atomic statements) that describe possible program behaviours.
%As such, a Petri program (like a \proglang program) contains its own specification:
%Any trace that labels a sequence of transitions leading to an accepting place is an \emph{error trace}, representing a specification violation.
%\medskip
%
%We will sometimes use an infix notation for the flow relation and write e.g. $p\mathrel{F}t$ instead of $(p,t)\in F$.
\goodbreak

Formally, we define a \emph{marking} as a map $m:P\to\mathbb{N}$ that assigns a token count to each place.
%A marking is \emph{accepting} if $m(p) > 0$ for some accepting place $p \in P_\mathsf{fin}$.
%We write $M$ to denote that set of all markings over $P$.
% The set of all markings over $P$ is $M_P$.
%We may write just $M$ if $P$ is known from the context.
%A marking $m\in M$ \emph{covers} a place $p\in P$ iff $m$ assigns at least
%one token to $p$.
%$$
%	m \text{ covers } p \Leftrightarrow m(p)>0
%$$
%We call a marking $m\in M$ \emph{accepting} iff it covers at least one accepting place.\smallskip
%
With $m\fire{\lambda(t)}m'$ we denote that transition $t\in T$ with the label $\lambda(t)$ can be \emph{fired} from marking $m$,
i.e., all predecessor places have a token -- formally, $m(p) > 0$ for all $p$ with $(p,t)\in F$ --,
and the firing of $t$ results in the marking $m'$
-- formally, $m'(p)=m(p)- \chi_F(\langle p,t \rangle) + \chi_F(\langle t,p\rangle)$, where $\chi_F$ is the characteristic function of $F$.
%Formally, we define the \emph{firing relation} $\fire{}\subseteq M\times T\times M$ as
%\[
%	m\fire{t}m'
%	~\Leftrightarrow~
%	\begin{array}{l}
%		\forall p\in P: (p, t) \in F \to m(p)>0 \qquad \text{and} \\
%		\forall p\in P: m'(p)=m(p)-|\{t\in T\mid (p,t)\in F\}|+|\{t\in T\mid (t,p)\in F\}|
%	\end{array}
%\]
%
%A trace $\tau$ is \emph{accepted} by the Petri program
%if there exists a so-called \emph{firing sequence} $m_0 \fire{t_1} m_1 \fire{t_2} \ldots \fire{t_n} m_n$,
%where $m_0$ is the initial marking, $m_n$ is accepting and $\tau=\lambda(t_1)\ldots\lambda(t_n)$.
A \emph{firing sequence} is a sequence $m_0 \fire{\lambda(t_1)} m_1 \fire{\lambda(t_2)} \ldots \fire{\lambda(t_n)} m_n$,
where $m_0 = m_\init$ is the initial marking.
%The language of the Petri program, $L(\N)$, is the set of all accepted traces.
% in $\N$ is then an alternating sequence $m_0\fire{t_1}m_1\fire{t_2}\ldots\fire{t_n}m_n$ of markings $m_i\in M$ and transitions $t_i \in T$,
%such that $(i)$ $m_0 = m_\mathsf{init}$ is the initial marking and $(ii)$ the sequence adheres to the firing relation, i.e. $m_{i-1}\fire{t_i}m_i$ for all $i\in\{1,\ldots,n\}$.
%A firing sequence ending in an accepting marking is called \emph{accepting}.
%
We say that a marking $m$ is reachable iff there exists a firing sequence $m_0 \fire{\lambda(t_1)} m_1 \fire{\lambda(t_2)} \ldots \fire{\lambda(t_n)} m_n$ with
%$m_0$ the initial marking and
$m_n = m$.
Analogously to~\cite{vmcai2021}, we only consider Petri programs that are \emph{1-safe}, i.e., where all reachable markings have at most one token per place.
%\smallskip
%
% We define reachability for markings and transitions
% \begin{align*}
% 	\text{marking } m\in M \text{ is reachable}
% 	&~\Leftrightarrow~
% 	\exists~ \text{firing seq. } m_0\fire{t_1}m_1\fire{t_2}\ldots\fire{t_n}m_n : m_n = m
% 	\\
% 	\text{transition } t\in T \text{ is reachable}
% 	&~\Leftrightarrow~
% 	\exists~ \text{firing seq. } m_0\fire{t_1}m_1\fire{t_2}\ldots\fire{t_n}m_n : t_n = t
% \end{align*}
%
%We define the language that is recognized by a Petri net as follows:
%\[
%	L(\N) := \left\{a_1a_2\ldots a_n\in\Sigma^*\mid
%	\begin{array}{l}
%		\exists~ \text{accepting firing sequence }\\\qquad m_0\fire{t_1}m_1\fire{t_2}\ldots\fire{t_n}m_n \\
%		\text{ such that }\forall~ i\in  \{1,\ldots n\} : \lab(t_i) = a_i
%	\end{array}
%	\right\}
%\]
%
%A net is \emph{bounded} (also known as \emph{1-safe} or just \emph{safe})
%iff all reachable markings have at most one token per place.
%In this paper we consider only bounded Petri nets and we often use \emph{Petri net} as a synonym for \emph{bounded Petri net}.
We thus identify reachable markings $m : P \to \{0,1\}$
%$m:P\to\mathbb{N}$
with sets of places%
% $m'\subseteq P$
.
%\[
%	m \equiv m'
%	\quad\Leftrightarrow\quad
%	\forall p\in P: m(p)=1 \leftrightarrow p\in m'
%\]
\medskip

Unlike \proglang programs, Petri programs do not have a notion of global and local variables.
The semantics are formulated over (global) program states $\sigma\in\petristate$, i.e., mappings from variables to their (boolean or integer) values.
Each atomic statement $\st$ has a semantic transition relation $\sem{\st} \subseteq \petristate\times\petristate$:
\begin{align*}
  \sem{\texttt{x:=e}} &:= \{\, (\sigma,\sigma') \mid \sigma' = \sigma[x\mapsto \sem{e}^\sigma] \,\}\\
  \sem{\texttt{assume e}} &:= \{\, (\sigma,\sigma') \mid \sigma = \sigma' \land \sem{e}^\sigma = \true \,\}
\end{align*}
%
%This semantic relation naturally extends to traces ($\sem{\varepsilon} = \mathit{id}$, $\sem{\st\,\tau} = \sem{\st}\circ \sem{\tau}$).
%Since accepted traces represent property violation,
%the verification of a Petri program amounts to showing that no such trace can (according to its semantics) be executed:
%
%
A \emph{specification} for a Petri program is a set of ``bad'' places $\spec$,
i.e., places that should not be reached by an execution of the Petri program.
%
%Thus, a
%
%
\begin{definition}[Satisfaction]
  A \emph{counterexample} to a specification $\spec$
  consists of a firing sequence $m_0\fire{\lambda(t_1)}\ldots\fire{\lambda(t_n)}m_n$
  and a sequence of states $\sigma_0, \ldots, \sigma_n \in \petristate$,
  such that $m_n \cap \spec \neq \emptyset$ (i.e., a bad place is reached),
  and we have $(\sigma_{i-1},\sigma_i)\in\sem{\lambda(t_i)}$ for all $i\in\{1,\ldots,n\}$
  (i.e., the trace corresponding to the firing sequence can actually be executed).

  The Petri program $\N$ \emph{satisfies} the specification $\spec$, denoted $\N \models \spec$,
  if there does not exist a counterexample to $\spec$ in $\N$.
\end{definition}

\subsection{Transformation to Petri Programs}
In this section, we present the transformation that we call \emph{petrification}.
The input for this construction consists of a \proglang program $\prog = (\body{}, \texttt{main}, \globalVars)$
and a \emph{thread limit} $\beta\in\mathbb{N}$.
The constructed Petri program represents all executions of the program $\prog$ where, at any time, at most $\beta$ threads with the same template are active,
i.e., executions whose thread width is at most $\beta$.

First, recall that, unlike \proglang programs,
Petri programs
%In our control flow we
do not have a notion of thread-local variables; all variables are global.
We thus rename variables of the program $\prog$ to ensure uniqueness.
Let $\vars$ be all the variables mentioned in $\prog$.
%, and as always, let $\globalVars$ be the global variables.
%Our control flow will range over the variables
We define the set \emph{instantiated variables} as follows:
\begin{multline*}
  \vars_\mathrm{inst} := \big\{\, \instvar{$x$}{\vartheta}{k} \mid x\in\vars\setminus\globalVars, k \in \{\bot,1,\ldots,\beta\},  \vartheta\in\threadTemplates \,\big\} \\
  \null \cup \globalVars \cup \big\{\, \idvar{\vartheta}{k} \mid k\in\{1,\ldots,\beta\}, \vartheta\in\threadTemplates \,\big\}
\end{multline*}
For each local variable $x$,
we define instantiated variables $\instvar{$x$}{\vartheta}{k}$,
where $\vartheta$ is a thread template, and $k\in\{\bot,1,\ldots,\beta\}$ is a unique \emph{instance ID}.
These IDs range from $1$ to $\beta$, with the special $\bot$ for the \texttt{main} thread that is initially active.
Global variables are not instantiated.
We also introduce variables \idvar{\vartheta}{k},
which keep track of the (non-unique) thread IDs used by \texttt{fork} and \texttt{join} statements.
We extend the idea of instantiation to expressions and atomic statements:
\instvar{$e$}{\vartheta}{k} denotes the expression derived by replacing every local variable $x\in\vars\setminus\globalVars$ in $e$ with \instvar{$x$}{\vartheta}{k}.
Similarly, the instantiated statement $[\st]_{\vartheta,k}$ is derived by replacing every local variable $x\in\vars\setminus\globalVars$ in the atomic statement $\st$ with \instvar{$x$}{\vartheta}{k}.
%For a simple statement $\st$, we define the instantiated statement
%%
%  \[
%      [\st]_{\vartheta,k} =
%      \begin{cases}
%        \texttt{assume \instvar{$e$}{\vartheta}{k}}
%          & \textbf{if } \st=\texttt{assume\;$e$}\\
%        \texttt{\instvar{$x$}{\vartheta}{k}:=\instvar{$e$}{\vartheta}{k}}
%          & \textbf{if } \st=\texttt{$x$:=$e$}\\
%        \texttt{\idvar{\vartheta}{k}:=\instvar{$e$}{\vartheta}{k}}
%          & \textbf{if } \st=\texttt{fork\;$e$\;$\vartheta$()}\\
%        \texttt{assume \idvar{\vartheta}{k}==\instvar{$e$}{\vartheta}{k}}
%          & \textbf{if } \st=\texttt{join\;$e$}
%      \end{cases}
%   \]
%%
%where \instvar{$e$}{\vartheta}{k} denotes the expression derived by replacing every local variable $x\in\vars\setminus\globalVars$ in $e$ with \instvar{$x$}{\vartheta}{k}.
%%  $e$, where  is replaced by .
%For \texttt{fork} and \texttt{join} statements, the concurrent aspect of their semantics is reflected in the structure of the Petri program.
%Hence, the instantiated statement merely takes care of the thread ID.
%Every step $(M,g)\semtrans{\st}(M',g')$ in an execution of $\prog$
%will be reflected in the Petri program by a transition labeled with the instantiated statement $[\st]_{\vartheta,k}$.

We introduce a formalism to capture the program control flow.
%This formalism uses a thread limit $\beta \in \mathbb{N}$ with the meaning that we only consider $\beta$ instances that are concurrently active.
First, let us define the set of \emph{control locations} $\mathbb{L}^\prog_\beta$:
%for a given program $\prog$ and thread limit $\beta$:
%
\begin{multline*}
\mathbb{L}^\prog_\beta = (\CommandPlus) \times \threadTemplates \times \{\bot,1,\ldots,\beta\} \\
  \cup \{\, \inuse{\vartheta}{k}, \notinuse{\vartheta}{k}, \insuff{\vartheta} \mid k \in \{1,\ldots,\beta\}, \vartheta \in \threadTemplates \,\}
\end{multline*}
There are several types of control locations:
\emph{Program control locations} consist of a command to be executed,
as well as a thread template and an instance ID that indicate to which thread the control location belongs.
%These program control locations consist of commands that can be executed that indicate the program locations, thread templates and unique instance IDs $\in \{\bot,1,\ldots,\beta\}$.
%Now we can extend this definition to the control locations $\mathbb{L}$
%  \[
%    \mathbb{L}_\beta = \mathbb{L^\prog_\beta} \cup \{\, \inuse{\vartheta}{k}, \notinuse{\vartheta}{k}, \insuff{\vartheta} \mid k \in \{1,\ldots,\beta\}, \vartheta \in \threadTemplates \,\}
%  \]
%
%These control locations consist of the program control locations and the
Additional control locations $\inuse{\vartheta}{k}$ and $\notinuse{\vartheta}{k}$
indicate if the thread with template $\vartheta$ and instance ID $k$ is currently active or not.
Finally, the control location $\insuff{\vartheta}$ indicates that the thread limit $\beta$ is insufficient for the program $\prog$,
specifically because more than $\beta$ threads with template $\vartheta$ can be created.
%
%The initial control locations $\mathbb{L}^\init_\beta \subseteq \mathbb{L}_\beta$ are those that indicate that the \texttt{main} thread is in its initial location (i.e., it has to execute the complete body) and all other threads are not active.
%  \[
%    \mathbb{L}^\init_\beta = \{\, \langle \body{\main}, \main, \bot \rangle\} \cup \{\notinuse{\vartheta}{k} \mid k \in \{1,\ldots,\beta\}, \vartheta \in \threadTemplates \,\}
%  \]

%Based on those variables,
We define the control flow as a ternary relation ${\petrirel{\cdot}}$ between sets of control locations, (instantiated) simple statements, and sets of control locations.
Specifically, let ${\petrirel{\cdot}}$ be the smallest relation such that the following conditions hold:
\begin{enumerate}
\item Let us fix a command $C$, some $X\in\CommandPlus$, a thread template $\vartheta$ and a simple statement $\st$.
  For every rule of the semantics definition (see \cref{fig:sem-def})
  that has the form
  \[
    \infer{\lbag\langle C, \vartheta, t, s \rangle\rbag, g \semtrans{\st} \lbag\langle X, \vartheta, t,  s'\rangle\rbag, g'}{\varphi} \ ,
  \]
  where $\varphi$ is only a side condition (i.e., $\varphi$ does not refer to the semantic transition relation),
  it holds that
  \begin{equation}
    \langle C, \vartheta, k \rangle \petrirel{[\st]_{\vartheta,k}} \langle X, \vartheta, k \rangle
    \label{eq:petrify-local}
  \end{equation}
  for all instance IDs $k \in \{\bot,1,\ldots,n\}$.
  In particular, this applies for the semantic rules
  \srule{Assume}, \srule{AssignGlobal}, \srule{AssignLocal}, \srule{Assert1}, \srule{Assert2}, \srule{Ite1}, \srule{Ite2}, \srule{While1} and \srule{While2}.
\item The following holds for all $X \in \Command \cup \{\Omega\}$, $k\in\{\bot,1,\ldots,\beta\}$ and $k'\in\{1,\ldots,\beta\}$:
\begingroup
\allowdisplaybreaks
  \begin{gather}
    \begin{multlined}
    \langle \texttt{fork\;$e$\;$\vartheta'$()}; X,\vartheta, k \rangle, \inuse{\vartheta'}{1},\ldots,\inuse{\vartheta'}{k'-1},\notinuse{\vartheta'}{k'}
    \qquad\qquad\\\qquad\qquad
    \petrirel{\texttt{\idvar{\vartheta'}{k'}:=\instvar{$e$}{\vartheta}{k}}} \langle X,\vartheta, k \rangle, \langle \body{\vartheta'},\vartheta', k' \rangle, \inuse{\vartheta'}{1},\ldots,\inuse{\vartheta'}{k'}
    \end{multlined}
    \label{eq:petrify-fork}
    \\[0.5em]
    \langle \texttt{fork\;$e$\;$\vartheta'$()}; X,\vartheta, k \rangle, \inuse{\vartheta'}{1},\ldots,\inuse{\vartheta'}{\beta} \petrirel{\texttt{assume true}} \insuff{\vartheta'}
    \label{eq:petrify-insuff}
    \\
    \langle \texttt{join\;$e$}; X,\vartheta, k \rangle, \langle \Omega, \vartheta',k' \rangle, \inuse{\vartheta'}{k'} \petrirel{\texttt{assume \idvar{\vartheta'}{k'}==\instvar{$e$}{\vartheta}{k}}} \langle X,\vartheta, k \rangle, \notinuse{\vartheta'}{k'}
    \label{eq:petrify-join}
  \end{gather}
\endgroup
%
%
%\item And finally, the relation satisfies the following inference rule:
%%
%  \begin{equation}
%    \infer{\{\langle C_1;C_2,\vartheta, k \rangle\} \uplus Z_1 \petrirel{\st} \{\langle X;C_2, \vartheta,k \rangle\} \uplus Z_2}{
%        \{\langle C_1,\vartheta,k \rangle\} \uplus Z_1 \petrirel{\st} \{ \langle X,\vartheta,k\rangle\} \uplus Z_2
%      }
%    \label{eq:petrify-seq}
%  \end{equation}
%%
%  $X$ here ranges over $\CommandPlus$.
%  Recall that we defined $\Omega;C_2=C_2$ and $\lightning;C_2=\lightning$.
%
%  three inference rules:
%  \begin{gather*}
%    \infer{\langle C_1;C_2,\vartheta, i \rangle \uplus Z_1 \petrirel{\st} \langle C_1';C_2, \vartheta,i \rangle \uplus Z_2}{
%        \langle C_1,\vartheta,i \rangle \uplus Z_1 \petrirel{\st} \langle C_1',\vartheta,i\rangle \uplus Z_2
%      }\\[0.2em]
%      %
%      \infer{\langle C_1;C_2,\vartheta, i \rangle \uplus Z_1 \petrirel{\st} \langle C_2,\vartheta, i \rangle \uplus Z_2}{
%          \langle C_1,\vartheta,i \rangle \uplus Z_1 \petrirel{\st} \langle \Omega,\vartheta,i\rangle \uplus Z_2
%        }
%        %\\[0.2em]
%        \quad\hfill\quad
%      %
%      \infer{\langle C_1;C_2,\vartheta, i \rangle \uplus Z_1 \petrirel{\st} \langle \lightning,\vartheta, i \rangle \uplus Z_2}{
%          \langle C_1,\vartheta,i \rangle \uplus Z_1 \petrirel{\st} \langle \lightning,\vartheta,i\rangle \uplus Z_2
%        }
%  \end{gather*}
\end{enumerate}
%
%\todo{
%  In our implementation, the insufficient threads transition seems to have the inUse-places as successors.
%  Is this necessary? Is this a good idea?
%%
%  Intuitively, the version defined above would take all tokens out of inUse places (meaning there are no tokens in notInUse-places either). Thus the the thread instance can never be forked again, the same insufficient-thread transition can never fire again, and 1-safety should hold.
%}
%
\Cref{eq:petrify-local} captures the sequential control flow within a thread,
in analogy to the \proglang semantics.
\Cref{eq:petrify-fork,eq:petrify-insuff,eq:petrify-join} implement the dynamic thread management and the thread limit $\beta$.
When a \texttt{fork} statement is executed,
the newly created thread is assigned a currently inactive instance ID $k'$.
Specifically, we always assign the minimal available instance ID.
If all instance IDs are already active, control gets stuck in the control location $\insuff{\vartheta'}$.
When a \texttt{join} statement is executed, the joined thread must have terminated (the remainder program is $\Omega$),
and the corresponding instance ID $k'$ is marked as inactive.
The full definition of $\petrirel{\cdot}$ \ifarxiv is shown in \cref{fig:full-petrification-trans} in the appendix.\else can be found in the extended version of the paper.\todo{cite} \fi
We can now define petrification:

\begin{definition}[Petrified Program]
    We define the \emph{petrified program}
    %for a program $\prog$ and a thread limit $\beta \in \mathbb{N}$
    as the Petri program
    $\petriProg{\beta}{\prog} := (P, T, F, m_\init, \lambda)$,
    where
    \begin{itemize}
      \item the places are control locations: $P=\mathbb{L}_\beta^\prog$,
      \item the transitions are triples defined by the control flow relation $\petrirel{\cdot}$:
      \[
        T = \{\, (\mathit{pred}, \st, \mathit{succ}) \mid \mathit{pred} \petrirel{\st} \mathit{succ} \,\},
      \]
      \item the flow relation follows directly from the transition triples:
      \[
        F = \{\, (p, (\mathit{pred},\st, \mathit{succ})) \mid p \in \mathit{pred} \,\} \cup \{\, ((\mathit{pred},\st, \mathit{succ}),p') \mid p' \in \mathit{succ} \,\},
      \]
      \item the initial marking consists of those control locations that indicate that the \main thread is in its initial location (i.e., it has to execute the complete body) and all other threads are not active
        \[ m_\init = \{\langle \body{\main}, \main, \bot \rangle\} \cup \{\notinuse{\vartheta}{k} \mid k \in \{1,\ldots,\beta\}, \vartheta \in \threadTemplates\} \]
      \item and the labeling of a transition triple is given by its second component:
      \[
        \lambda\big( (\mathit{pred},\st, \mathit{succ}) \big) = \st.
      \]
    \end{itemize}
\end{definition}
Note that $\mathbb{L}_\beta^\prog$ has infinitely many elements, because there are infinitely many commands.
However, only a finite number of control locations are actually relevant to the program.
We omit certain unreachable places and transitions in $\petriProg{\beta}{\prog}$, such that it only has finitely many places.
%\todo{Do we need an example for this?}
%

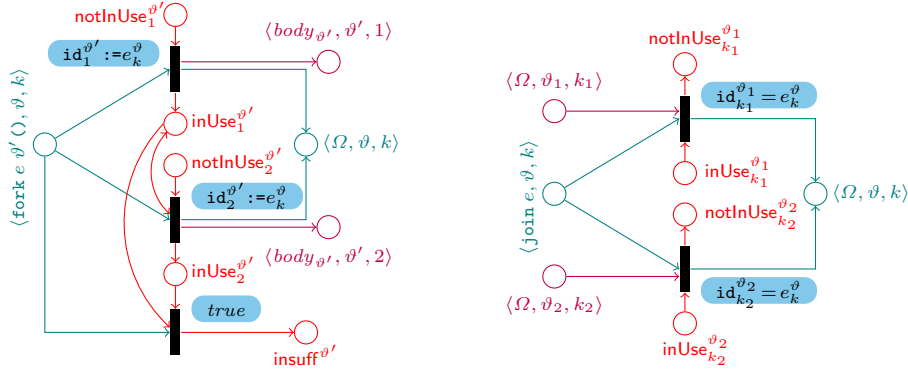
\begin{figure}[t]
\tikzstyle{pl}=[draw,circle]
\begin{tikzpicture}
  \node[pl,draw=teal,label={[above,rotate=90,xshift=-2mm,yshift=1mm,teal,font=\scriptsize]$\langle \texttt{fork\;$e$\;$\vartheta'$()}, \vartheta, k\rangle$}] (l11) {};
  \node[transition,right=15mm of l11,yshift= 10mm,label={[left,yshift=-1mm,xshift=-1mm]\stsmcol{\scriptsize\idvar{\vartheta'}{1}:=\instvar{$e$}{\vartheta}{k}}{th1}}] (fork1) {};
  \node[transition,right=15mm of l11,yshift=-10mm,label={[right,xshift=1mm]\stsmcol{\scriptsize\idvar{\vartheta'}{2}:=\instvar{$e$}{\vartheta}{k}}{th1}}] (fork2) {};
  \node[transition,right=15mm of l11,yshift=-25mm,label={[right,xshift=1mm]\stsmcol{\scriptsize$\true$}{th1}}] (fork3) {};
  \node[pl,draw=teal,  right=15mm of fork1,yshift=-10mm,label={[right,yshift=-1.5mm,xshift=1mm,teal,font=\scriptsize]$\langle \Omega, \vartheta, k\rangle$}] (l12) {};
  \node[pl,draw=purple,right=18mm of fork1,yshift=1mm,label={[above,purple,font=\scriptsize]$\langle \body{\vartheta'}, \vartheta', 1\rangle$}] (l21a) {};
  \node[pl,draw=purple,right=18mm of fork2,yshift=-1mm,label={[below,yshift=-3mm,purple,font=\scriptsize]$\langle \body{\vartheta'}, \vartheta', 2\rangle$}] (l21b) {};

  \node[pl,above= 2.5mm of fork1,red,label={[left,yshift=-1.5mm,red,font=\scriptsize]$\notinuse{\vartheta'}{1}$}] (niu1) {};
  \node[pl,below= 2.5mm of fork1,red,label={[right,xshift=1mm,yshift=-1mm,red,font=\scriptsize]$\inuse{\vartheta'}{1}$}]    (iu1)  {};
  \node[pl,above= 2.5mm of fork2,red,label={[right,xshift=1mm,yshift=-1mm,red,font=\scriptsize]$\notinuse{\vartheta'}{2}$}] (niu2) {};
  \node[pl,below= 2.5mm of fork2,red,label={[right,xshift=1mm,yshift=-1mm,red,font=\scriptsize]$\inuse{\vartheta'}{2}$}]    (iu2)  {};
  \node[pl,right=15mm of fork3,red,label={[below,yshift=-2mm,red,font=\scriptsize]$\insuff{\vartheta'}$}]    (ins)  {};

  \draw[->,teal]   (l11)   -- (fork1.west);
  \draw[->,red]    (niu1)  -- (fork1);
  \draw[->,teal]   (fork1) -| (l12);
  \draw[->,purple] (fork1.east) ++(up:1mm) -- (l21a);
  \draw[->,red]    (fork1) -- (iu1);

  \draw[->,teal]   (l11)   -- (fork2.west);
  \draw[<->,red]    (iu1) edge[bend right=45] (fork2);
  \draw[->,red]    (niu2)  -- (fork2);
  \draw[->,teal]   (fork2) -| (l12);
  \draw[->,purple] (fork2.east) ++(down:1mm) -- (l21b);
  \draw[->,red]    (fork2) -- (iu2);

  \draw[->,teal]   (l11)   |- (fork3.west);
  \draw[red]  (iu1.west) edge[->,bend right=45] (fork3);
  \draw[->,red]    (iu2) -- (fork3);
  \draw[->,red]    (fork3) -- (ins);
\end{tikzpicture}
\hfill
\begin{tikzpicture}
  \node[pl,draw=teal,label={[above,rotate=90,xshift=-2mm,yshift=1mm,teal,font=\scriptsize]$\langle \texttt{join\;$e$}, \vartheta, k\rangle$}] (l11) {};
  \node[transition,right=15mm of l11,yshift= 10mm,label={[right,xshift=1mm]\stsmcol{\scriptsize$\idvar{\vartheta_1}{k_1}\!=\!\instvar{$e$}{\vartheta}{k}$}{th1}}] (join1) {};
  \node[transition,right=15mm of l11,yshift=-10mm,label={[right,xshift=1mm,yshift=-6mm]\stsmcol{\scriptsize$\idvar{\vartheta_2}{k_2}\!=\!\instvar{$e$}{\vartheta}{k}$}{th1}}] (join2) {};
  \node[pl,draw=teal,  right=15mm of join1,yshift=-10mm,label={[right,yshift=-1.5mm,xshift=1mm,teal,font=\scriptsize]$\langle \Omega, \vartheta, k\rangle$}] (l12) {};
  \node[pl,draw=purple,left=15mm of join1,yshift=1mm,label={[above,purple,font=\scriptsize]$\langle \Omega, \vartheta_1, k_1\rangle$}] (l21a) {};
  \node[pl,draw=purple,left=15mm of join2,yshift=-1mm,label={[below,yshift=-3mm,purple,font=\scriptsize]$\langle \Omega, \vartheta_2, k_2\rangle$}] (l21b) {};

  \node[pl,above= 2.5mm of join1,red,label={[above,yshift=-1.5mm,xshift=1.5mm,red,font=\scriptsize]$\notinuse{\vartheta_1}{k_1}$}] (niu1) {};
  \node[pl,below= 2.5mm of join1,red,label={[right,xshift=1.5mm,yshift=-1.5mm,red,font=\scriptsize]$\inuse{\vartheta_1}{k_1}$}]    (iu1)  {};
  \node[pl,above= 2.5mm of join2,red,label={[right,xshift=1.5mm,yshift=-1.5mm,red,font=\scriptsize]$\notinuse{\vartheta_2}{k_2}$}] (niu2) {};
  \node[pl,below= 2.5mm of join2,red,label={[below,yshift=-2mm,xshift=1.5mm,red,font=\scriptsize]$\inuse{\vartheta_2}{k_2}$}]    (iu2)  {};

  \draw[->,teal]   (l11)   -- (join1.west);
  \draw[<-,red]    (niu1)  -- (join1);
  \draw[->,teal]   (join1) -| (l12);
  \draw[<-,purple] (join1.west) ++(up:1mm) -- (l21a);
  \draw[<-,red]    (join1) -- (iu1);

  \draw[->,teal]   (l11)   -- (join2.west);
  \draw[<-,red]    (niu2)  -- (join2);
  \draw[->,teal]   (join2) -| (l12);
  \draw[<-,purple] (join2.west) ++(down:1mm) -- (l21b);
  \draw[<-,red]    (join2) -- (iu2);
\end{tikzpicture}
\caption{Illustration of the transitions derived from \cref{eq:petrify-fork,eq:petrify-insuff} (on the left) and from \cref{eq:petrify-join} (on the right), for $\beta=2$.}
\label{fig:petrify-fork-join}
\end{figure}
\todo{Fix the overlapping labels in \cref{fig:petrify-fork-join}}

\cref{fig:petrify-fork-join} illustrates the transitions created for fork and join commands in the petrified program (for $\beta=2$).
The three transitions on the left correspond to \texttt{fork\;$e$\;$\vartheta$()}.
The uppermost transition creates the thread with instance ID 1.
It can only be fired if the thread with instance ID 1 is not yet active (there is a token in $\notinuse{\vartheta}{1}$).
After firing the transition, the thread with instance ID 1 is active (there is a token in $\inuse{\vartheta}{1}$), its thread ID was stored in \idvar{\vartheta}{1}, and the command $\body{\vartheta}$ remains to be executed.
%
%The transition in the middle creates the thread with instance ID 2.
%It can only be fired if the thread with instance ID is already active,
%but the thread with instance ID 2 is not (i.e., there are tokens in $\inuse{\vartheta}{1}$ and $\notinuse{\vartheta}{2}$).
%After firing the transition, both threads are active (i.e., there are tokens in $\inuse{\vartheta}{1}$ and $\inuse{\vartheta}{2}$),
%its thread ID was stored in \idvar{\vartheta}{2}, and the command $\body{\vartheta}$ remains to be executed.
%
The transition in the middle behaves similarly, for the thread with instance ID 2.
However, it can only be executed if the thread with instance ID 1 is already active.
The transition at the bottom indicates that our chosen thread limit $\beta$ was not sufficient. This transition can be fired if both thread instances are active (there are tokens in $\inuse{\vartheta}{1}$ and $\inuse{\vartheta}{2}$).

To the right of \cref{fig:petrify-fork-join},
we see the transitions generated for a \texttt{join} command.
%On the right there is the functionality of the join transitions illustrated.
For each thread instance (regardless of the template),
there is a transition that corresponds to \texttt{join\;$e$}.
This transition can be fired if the thread instance is active (i.e., there is a token in the corresponding $\inuse{}{}$ place).
After firing the transition, the thread instance is no longer active (i.e., there is a token in the corresponding $\notinuse{}{}$ place).
The transition label ensures that the thread ID has the value of \texttt{e} (we omit the keyword \texttt{assume} for brevity).

\begin{example}
    The Petri net in \cref{fig:petri-net} represents $\petriProg{2}{\prog}$, where $\prog$ is the program shown in \cref{fig:example}.
    For readability, we forgo annotating the places with the corresponding program control locations.
    %We just did not use triples of $\Command$, $\threadTemplates$ and instance IDs as place-names for convenience reasons.
    For example, the place named $\ell_{0,1}$ represents the control location $\langle \texttt{c:=c+i\,;\,assert c<=2*i\,;\,c:=c-i}, \texttt{w}, 1 \rangle$.
\end{example}

\subsection{Properties of the Petrified Program}
We show that the petrified program $\petriProg{\beta}{\prog}$ for a given program $\prog$ satisfies certain properties,
which allow it to be used in the verification of $\prog$.
Detailed proofs of these results can be found in \ifarxiv the appendix. \else the extended version of the paper.\todo{cite} \fi

Recall that the  Petri program verification algorithm~\cite{vmcai2021} only supports 1-safe Petri programs.
We thus show that petrification ensures 1-safety.
% can only be applied for 1-safe petri-programs. Since we want to apply this algorithm to our construction $\petriProg{\beta}{\prog}$, we need to show that it satisfies this property.
%
To this end, observe that all reachable markings of a petrified program $\petriProg{\beta}{\prog}$
satisfy the following conditions,
%have a particular form.
%Specifically, the following holds
for all $\vartheta\in\threadTemplates$ and $k\in\{1,\ldots,\beta\}, k'\in\{\bot,1,\ldots,\beta\}$:
\begin{itemize}
  \item The sum of the tokens in $\inuse{\vartheta}{k}$, $\notinuse{\vartheta}{k}$ and $\insuff{\vartheta}$ is exactly 1.
  % There is exactly one token between the places
  \item The place $\inuse{\vartheta}{k}$ has a token iff there exists some $X \in \CommandPlus$ such that $\langle X,\vartheta, k\rangle$ has a token.
  \item %There is at most one token between all places of the form $\langle X,\vartheta,k'\rangle$ with some $X \in \CommandPlus$.
    The sum of the tokens in all places of the form $\langle X,\vartheta,k'\rangle$ (with some $X \in \CommandPlus$) is at most 1.
\end{itemize}
We call a marking that satisfies these conditions \emph{coherent}.
\begin{toappendix}
\begin{definition}[Coherence]
  A marking $m$ of the petrified program $\petriProg{\beta}{\prog}$ is \emph{coherent} if the following conditions hold for all $\vartheta\in\threadTemplates$ and $k\in\{1,\ldots,\beta\}, k'\in\{\bot,1,\ldots,\beta\}$:
  \begin{gather}
    m(\inuse{\vartheta}{k}) + m(\notinuse{\vartheta}{k}) + m(\insuff{\vartheta}) = 1\\
    m(\inuse{\vartheta}{k}) > 0 \iff \exists X \in \CommandPlus\,.\, m(\langle X,\vartheta, k\rangle) > 0\\
    \sum_{X\in\CommandPlus} m(\langle X,\vartheta,k'\rangle) \leq 1
  \end{gather}
\end{definition}
\end{toappendix}
\begin{lemma}
  All reachable markings of $\petriProg{\beta}{\prog}$ are coherent.
\end{lemma}
\begin{proof}
  It is easy to see that the initial marking is coherent.
  Furthermore, observe that the successor of a coherent marking is again coherent:
  Each transition according to \cref{eq:petrify-local,eq:petrify-fork,eq:petrify-insuff,eq:petrify-join} preserves coherence.
\end{proof}
%
%Coherence directly implies 1-safety:
%
\begin{proposition}[1-Safety]
  \label{prop:petrify-safe}
  The Petri program $\petriProg{\beta}{\prog}$ is 1-safe.
\end{proposition}
\begin{proof}
  Any coherent marking assigns at most one token to a place.
  Since all reachable markings are coherent, it follows that the Petri program is 1-safe.
\end{proof}

Thus we can verify the petrified program using an existing algorithm~\cite{vmcai2021}.
It remains to give a specification against which the petrified program shall be verified.
In fact, we define \emph{two} specifications for the petrified program $\petriProg{\beta}{\prog}$:
one that indicates the absence of assert violations,
and another that indicates sufficiency of the thread limit $\beta$.
\begin{definition}[Specifications]
  The \emph{safety specification} $\spec_\safe$ and the \emph{bound specification} $\spec_\bound$ of the petrified program $\petriProg{\beta}{\prog}$ are given by the sets of places
  \begin{align*}
    \spec_\safe &= \{\, \langle\lightning,\vartheta, k\rangle \mid k\in\{\bot,1,\ldots,\beta\} \land \vartheta \in \threadTemplates \,\}\\
    \spec_\bound &= \{\, \insuff{\vartheta} \mid \vartheta\in\threadTemplates\,\}
  \end{align*}
\end{definition}

%We define the thread bound-focused variant of the petrified program as follows:
%%
%\begin{definition}
%  Let $\petriProg{\beta}{\prog} = (\Sigma, P, T, F, m_\init, P_\mathsf{fin})$.
%  Then the thread bound-focused variant of the petrified program $\petriProg[\mathsf{bound}]{n}{\prog}$ is derived
%  by simply exchanging the accepting places:
%  \[
%    \petriProg[\mathsf{bound}]{n}{\prog} = (\Sigma, P, T, F, m_\init, \{\, \insuff{\vartheta} \mid \vartheta\in\threadTemplates\,\})
%  \]
%\end{definition}
%%
%Since the accepting places do not influence 1-safety, it follows that \cref{prop:petrify-safe} also holds for $\petriProg[\mathsf{bound}]{n}{\prog}$.
%
%In essence, we are giving the program a new specification:
%Namely, that $n$ thread instances are sufficient for the program.
%

It remains to be shown %for the application of the verification algorithm on $\petriProg{\beta}{\prog}$ is
that the specifications of the petrified program actually correspond to the behaviour of the program $\prog$.
In order to show this, we first create a link between executions of the program $\prog$
and firing sequences of the petrified program $\petriProg{\beta}{\prog}$.
In particular, we map a given marking $m$ and a state $\sigma$ %of the Petri program
to a corresponding global configuration:
We create local configurations for all program control locations in $m$, and extract the thread ID as well as local and global states from $\sigma$.
Formally, $\conf{m}{\sigma} := (M, \sigma|_\globalVars)$
where
\[
    M = \lbag\, \langle C, \vartheta, t, \lfloor\sigma\rfloor_{\vartheta,k}\rangle \mid \langle C, \vartheta, k \rangle \in m \land (t = k = \bot \lor t = \sigma(\idvar{\vartheta}{k})) \,\rbag
\]
with $\lfloor\sigma\rfloor_{\vartheta,k}(x) = \sigma(\instvar{$x$}{\vartheta}{k})$ for all local variables $x\in\vars\setminus\globalVars$.
\begin{toappendix}
  Throughout the proofs of the lemmas below, we need an extended version of the instantiation of atomic statements.
  Specifically, we extend instantiation from atomic statements to simple statements,
  i.e., to fork and join statements.
  For fork and join statements, we need two pairs of thread template and instance ID.
  We define, in analogy to \cref{eq:petrify-fork,eq:petrify-join}:
  \begin{align*}
    [\texttt{fork\;$e$\;$\vartheta'$()}]_{\vartheta,k}^{\hat\vartheta,\hat{k}} &= \texttt{\idvar{\hat\vartheta}{\hat{k}}:=\instvar{$e$}{\vartheta}{k}}\\
    [\texttt{join\;$e$}]_{\vartheta,k}^{\hat\vartheta,\hat{k}} &= \texttt{assume\;\idvar{\hat\vartheta}{\hat{k}}==\instvar{$e$}{\vartheta}{k}}\\
    [\st]_{\vartheta,k}^{\hat\vartheta,\hat{k}}  &= [\st]_{\vartheta,k} \qquad\qquad\textbf{for } \st\in\atomicstmt
  \end{align*}
  Furthermore, given a statement $\tilde{\st}$ from the petrified program,
  we define the \emph{de-instantiated statement} $\lfloor\tilde{\st}\rfloor$
  as the (unique) statement $\st$ such that $ \tilde{\st} = [\st]_{\vartheta,k}^{\vartheta',k'}$ for some $\vartheta,\vartheta',k,k'$.
\end{toappendix}
The following lemmata allow us to associate an execution with firing sequences that are executable according to Petri program semantics:
\begin{lemmarep}
\label{lem:fire-exec}
  Let $m_0 \fire{\tilde{\st_1}} \ldots \fire{\tilde{\st_n}} m_n$ be a firing sequence of $\petriProg{\beta}{\prog}$
  where none of the markings $m_i$ contains a place $\insuff{\vartheta}$ for any $\vartheta$.
  Let $\sigma_0,\ldots,\sigma_n$ be states with $(\sigma_{i-1},\sigma_i) \in \sem{\tilde{\st_i}}$ for all $i\in\{1,\ldots,n\}$.
  There exist simple statements $\st_1,\ldots,\st_n$ such that
 % \[
 $
    %\conf{m_0}{\sigma_0} \semtrans{\lfloor \tilde{\st_1} \rfloor} \ldots \semtrans{\lfloor \tilde{\st_n} \rfloor} \conf{m_n}{\sigma_n}
    \conf{m_0}{\sigma_0} \semtrans{\st_1} \ldots \semtrans{\st_n} \conf{m_n}{\sigma_n}
  %\]
  $
  is an execution.
\end{lemmarep}
\begin{proofsketch}
  We proceed by induction over $n$.
  For the base case, observe that the global configuration $\conf{m_\init}{\sigma_0}$ is initial.
  For the induction step, it remains to justify
  %show that the semantic transition relation satisfies
  %$\conf{m_{n-1}}{\sigma_{n-1}} \semtrans{\lfloor \tilde{\st_{n}} \rfloor} \conf{m_n}{\sigma_n}$.
  the last step in the execution, using the semantic relation.
  We proceed by case distinction over the rule (among \cref{eq:petrify-local,eq:petrify-fork,eq:petrify-join}) that is responsible for the last transition in the firing sequence
  (transitions according to \cref{eq:petrify-insuff} are ruled out by our assumption on the firing sequence).
  Using the semantic rules corresponding to each case, in combination with the \srule{Frame} rule,
  we show that the above semantic transition indeed exists.
\end{proofsketch}
\begin{proof}
  We choose the simple statements $\st_i := \lfloor \tilde{\st}_i \rfloor$, for all $i$.
  We proceed by induction over $n$.

  For the base case $n=0$,
  we only have to show that $\conf{m_0}{\sigma_0}$ is an initial global configuration.
  Since the firing sequence must begin with the initial marking,
  we have $m_0 = \mathbb{L}_\beta^\init$.
  It follows that
  \[
    \conf{m_0}{\sigma_0} = (\lbag \body{\main}, \main, \bot, \lfloor\sigma_0\rfloor_{\vartheta,\bot} \rbag, \sigma_0|_\globalVars),
  \]
  and hence $\conf{m_0}{\sigma_0}$ is indeed initial.

  For the inductive step, we assume that
  \[
    \conf{m_0}{\sigma_0} \semtrans{\lfloor \tilde{\st_1} \rfloor} \ldots \semtrans{\lfloor \tilde{\st_{n-1}} \rfloor} \conf{m_{n-1}}{\sigma_{n-1}}
  \]
  is an execution.
  It remains only to show that the semantic transition relation satisfies $\conf{m_{n-1}}{\sigma_{n-1}} \semtrans{\lfloor \tilde{\st_{n}} \rfloor} \conf{m_n}{\sigma_n}$.
  By the definition of the petrified program,
  we must have $m_{n-1} = m'_{n-1} \uplus \tilde{m}$ and $m_n = m'_n \uplus \tilde{m}$ such that $m'_{n-1} \petrirel{\tilde{\st}_{n}} m'_n$,
  for some suitable $m'_{n-1}$, $m'_n$ and $\tilde{m}$.
  We proceed by case distinction over the definition of the control flow relation $\petrirel{\cdot}$.
  In each case, we show that $\conf{m'_{n-1}}{\sigma_{n-1}} \semtrans{\lfloor \tilde{\st_{n}} \rfloor} \conf{m'_n}{\sigma_n}$.
  The result follows by application of the  semantic rule \srule{Frame}, where $M_2$ is the first component of $\conf{\tilde{m}}{\sigma_n}$.
  \begin{itemize}
    \item If the control flow transition $m'_{n-1} \petrirel{\tilde{\st}_{n}} m'_n$ is derived from \cref{eq:petrify-local},
      then we have $m'_{n-1} = \{ \langle C,\vartheta, k \rangle \}$ and $m'_n = \{\langle X,\vartheta,k\rangle\}$ for some $C,X,\vartheta$ and $k$.
      Further, we know that
      \[
        \infer{\lbag\langle C, \vartheta, t, s \rangle\rbag, g \semtrans{\st} \lbag\langle X, \vartheta, t,  s'\rangle\rbag, g'}{\varphi} \ ,
      \]
      is a rule in the semantics definition, where $\tilde{\st}_n = [\st]_{\vartheta,k}$, or equivalently, $\st = \lfloor \tilde{\st}_n \rfloor$.
      Examining each of the relevant semantic rules,
      we observe that in each case the side condition $\varphi$ follows from the fact that $(\sigma_{n-1}, \sigma_n) \in \sem{\tilde{\st}_n}$.
      Thus we conclude that  $\conf{m'_{n-1}}{\sigma_{n-1}} \semtrans{\lfloor \tilde{\st_{n}} \rfloor} \conf{m'_n}{\sigma_n}$.

    \item If the control flow transition $m'_{n-1} \petrirel{\tilde{\st}_{n}} m'_n$ is derived from \cref{eq:petrify-fork},
      then we have $m'_{n-1} = \{ \langle \texttt{fork $e$ $\vartheta'$()}, \vartheta, k \rangle, \inuse{\vartheta'}{1},\ldots,\inuse{\vartheta'}{k'-1},\notinuse{\vartheta'}{k'} \}$ and
      $m'_n = \{ \langle \Omega,\vartheta, k \rangle, \langle \body{\vartheta'},\vartheta', k' \rangle, \inuse{\vartheta'}{1},\ldots,\inuse{\vartheta'}{k'} \}$, for some $e,\vartheta,\vartheta',k,k'$.
      Further, we have that $\tilde{\st}_n$ is the statement \texttt{\idvar{\vartheta'}{k'}:=\instvar{$e$}{\vartheta}{k}},
      and thus $\lfloor\st_n\rfloor$ is the statement \texttt{fork $e$ $\vartheta$'()}.
      Using \srule{Fork}, it is straightforward to conclude that indeed  $\conf{m'_{n-1}}{\sigma_{n-1}} \semtrans{\lfloor \tilde{\st_{n}} \rfloor} \conf{m'_n}{\sigma_n}$.
    \item The control flow transition can not be derived from \cref{eq:petrify-insuff},
      since then we would have $\insuff{\vartheta'} \in m_n$, violating our assumption on the firing sequence.
    \item If the control flow transition $m'_{n-1} \petrirel{\tilde{\st}_{n}} m'_n$ is derived from \cref{eq:petrify-join},
      then we have
      \[
        m'_{n-1} = \{\langle \texttt{join $e$},\vartheta, k \rangle, \langle \Omega, \vartheta',k' \rangle, \inuse{\vartheta'}{k'}\}
      \]
      \[
        m'_n = \{\langle \Omega,\vartheta, k \rangle, \notinuse{\vartheta'}{k'}\}
      \]
      for some $e,\vartheta,\vartheta',k,k'$.
      Further, we have that $\tilde{\st}_n =\texttt{assume \idvar{\vartheta'}{k'}==\instvar{$e$}{\vartheta}{k}}$,
      and thus $\lfloor\st_n\rfloor$ is the statement \texttt{join $e$}.
      Using \srule{Join}, it is straightforward to conclude that indeed  $\conf{m'_{n-1}}{\sigma_{n-1}} \semtrans{\lfloor \tilde{\st_{n}} \rfloor} \conf{m'_n}{\sigma_n}$.
%    \item Finally, if the control flow transition $m'_{n-1} \petrirel{\tilde{\st}_{n}} m'_n$ is derived from \cref{eq:petrify-seq},
%      then we have $m'_{n-1} = \{\langle C_1;C_2, \vartheta, k \rangle\} \uplus Z_1$ and $m'_n = \{\langle X;C_2,\vartheta,k\rangle\} \uplus Z_2$.
%      By induction over the control flow relation, we conclude that
%      \[
%        \conf{\{\langle C_1, \vartheta, k \rangle\} \uplus Z_1}{\sigma_{n-1}} \semtrans{\lfloor\tilde{\st}_n\rfloor} \conf{\{\langle X,\vartheta,k\rangle\} \uplus Z_2}{\sigma_n}
%      \]
%      holds.
%      Using the semantic rule \srule{Seq}, we conclude that indeed  $\conf{m'_{n-1}}{\sigma_{n-1}} \semtrans{\lfloor \tilde{\st_{n}} \rfloor} \conf{m'_n}{\sigma_n}$.
  \end{itemize}
\end{proof}
%
%
%
%In order to complete the link between firing sequences of the petrified program
%and executions,
%we also need a lemma for the opposite direction:
%
%
%
\begin{toappendix}
  \subsection{Proof of \cref{lem:exec-fire}}

  For the proof of \cref{lem:exec-fire},
  we proceed in three steps.
  \begin{itemize}
  \item First, we augment the global configurations of the execution
    by assigning an instance ID $k\in\{1,\ldots,\beta\}$ to each local configuration.
    We further define two mappings $\alpha,\gamma:\{1,\ldots,n\} \to \{\bot,1\ldots,\beta\}$
    (where $n+1$ is the length of the execution),
    which indicate the threads involved in making the respective steps in the execution.
    We need two mappings to cover the cases of fork and join,
    where 2 threads are involved.
    We define a notion of \emph{admissibility} for such an augmented execution and the respective thread mappings.

  \item Second, we show inductively that for every execution, there exist an admissible augmented execution and two admissible thread mappings.

  \item Third, we show that from an admissible augmented execution and thread mappings,
    one can extract a corresponding firing sequence of the petrified program, along with the sequence of instantiated states.
  \end{itemize}

  An \emph{augmented local configuration} is a tuple $\langle C, \vartheta, t, s, k\rangle$,
  such that $\langle C, \vartheta, t,  s \rangle$ is a local configuration,
  and $k \in \{\bot,1,\ldots,\beta\}$ is an instance ID.
  An \emph{augmented global configuration} is a pair $(\aug{M}, g)$ such that $\aug{M}$ is a set (not a multiset) of augmented local configurations and $g$ is a global state.
  We call a set $\aug{M}$ of augmented local configurations \emph{conformist},
  if $\aug{M}$ contains no two augmented local configurations with the same thread template and instance ID,
  i.e., for all $\langle C,\vartheta,t,s,k \rangle, \langle C',\vartheta,t',s',k \rangle\in\aug{M}$, we have $C=C'$, $t=t'$, and $s=s'$.

  We use $\lfloor\cdot\rfloor$ do denote the \emph{deaugmentation},
  i.e., we define $\lfloor \langle C, \vartheta, t, s, k\rangle\rfloor := \langle C, \vartheta, t,  s\rangle$.
  We lift $\lfloor\cdot\rfloor$ to sets of augmented local configurations, such that it returns a multiset of local configurations:
  \[
    \lfloor \aug{M}\rfloor:=\lbag\, \langle C, \vartheta, t, s\rangle \mid \langle C, \vartheta, t, s, \vartheta, k\rangle\in \aug{M}\,\rbag
  \]

  Given a sequence of augmented global configurations $\aug{M}_0, \ldots, \aug{M}_n$,
  a \emph{thread mapping} $\alpha$ maps each number $i\in\{0\ldots n-1\}$
  to a pair $\langle \vartheta, k \rangle \in \threadTemplates \times \{\bot,1,\ldots,\beta\}$.
  We use two thread mappings:
  A \emph{current thread mapping} $\alpha$ that indicates which thread executes the next statement,
  and a \emph{secondary thread mapping} $\gamma$ that indicates, for \texttt{fork} resp.\ \texttt{join} statements,
  the forked resp.\ the joined thread.

  \begin{definition}[Awake]
    Let $\aug{M}_1, \aug{M}_2$ be sets of augmented local configurations,
    let $\occ \subseteq \threadTemplates \times \{1,\ldots,\beta\}$
    and let $\st$ be a simple statement.
    A quadruple $\langle \vartheta, k, \vartheta', k' \rangle$ is \emph{$(\aug{M}_1,\aug{M}_2,\occ,\st)$-awake}
    if the following holds:
    \begin{itemize}
      \item If $\st$ is an atomic statement,
           there exist $C,C',t,s,s'$ such that
            $\aug{M}_2=\aug{M}_1 \setminus \{ \langle C, \vartheta, t, s, k \rangle \}\cup\{\langle C',\vartheta, t, s', k\rangle\}$.
      \item If $\st$ is a fork statement of the form \textup{\texttt{fork\;$e$\;$\hat\vartheta$()}},
        then $\hat\vartheta = \vartheta'$,
        and there exist $C, C', t, t', s, s'$
        such that the following two equalities hold:
        \begin{align*}
          \aug{M}_2 & = \aug{M}_1 \setminus \{ \langle C, \vartheta, t, s, k \rangle \} \cup \{
                    \langle C', \vartheta, t, s, k \rangle,
                    \langle\body{\vartheta'}, \vartheta', t', s', k'\rangle
                  \}\\
          k' &= {\min} \{\, \tilde{k} \in \{1,\ldots,\beta\} \mid
            \langle \vartheta',\tilde{k}\rangle \notin \occ
            \land \lnot \exists \tilde{C},\tilde{t}, \tilde{s}\,.\, \langle\tilde{C}, \vartheta', \tilde{t}, \tilde{s},\tilde{k}\rangle \in \aug{M}_1 \,\}
        \end{align*}
      \item If $\st$ is a join statement, then there exist $C,t,t',s,s'$ such that
        \[
          \aug{M}_2 = \aug{M}_1 \setminus \{ \langle \Omega, \vartheta', t', s', k' \rangle, \langle C, \vartheta, t, s, k \rangle \}
          \cup \{\langle C', \vartheta, t, s, k \rangle\}
        \]
    \end{itemize}
  \end{definition}
  Intuitively, the idea is that the augmented local configuration of the awake thread (with template $\vartheta$ and instance ID $k$)
  keeps its instance ID when the next statement is executed,
  and the other augmented local configurations are not modified.
  If the statement is a \texttt{fork} statement, the augmented local configuration of the new thread (with template $\vartheta'$) takes the smallest instance ID $k'$
  that is not yet utilized as instance ID of another thread that has the same template.
  Additionally, threads in $\occ$ are ``occupied'', and may not be used by a \texttt{fork} statement (we need this in an induction proof below).
  For \texttt{join} statements, we require that the joined thread has template $\vartheta'$ and instance ID $k'$.

  \begin{definition}[Admissible]
    We call a triple that consists of a sequence of augmented global configurations $\aug{M}_0,\ldots,\aug{M}_n$,
    a current thread mapping $\alpha$,
    and a thread mapping $\gamma$ \emph{admissible} for an an execution $(M_0,g_0) \semtrans{\st_1} \ldots \semtrans{\st_n} (M_n, g_n)$ if the following holds.
    \begin{itemize}
      \item The instance ID of the augmented initial local configuration is $\bot$, i.e.,
      $$\aug{M}_0=\{ \langle C, \vartheta, t, s, \bot \rangle \}\ \text{ for some }C, \vartheta, t, s.$$
      \item For all $i\in\{0,\ldots n\}$,
      \begin{itemize}
        \item the deaugmentation of the augmented local configurations coincides with the corresponding local configurations, i.e.,
        $\lfloor\aug{M}_i\rfloor = M_i$
        \item and $\aug{M}_i$ is conformist.
      \end{itemize}
      \item For all  $i\in\{0,\ldots n-1\}$,
        if $\alpha(i)=\langle \vartheta,k\rangle$ and $\gamma(i) = \langle \vartheta',k'\rangle$,
        then the quadruple $\langle \vartheta,k,\vartheta',k' \rangle$
        is $(\aug{M}_i, \aug{M}_{i+1}, \emptyset, \st_{i+1})$-awake.
    \end{itemize}
  \end{definition}

  \begin{lemma}\label{lem:exists-admissible}
    For each execution $(M_0,g_0) \semtrans{\st_1} \ldots \semtrans{\st_n} (M_n, g_n)$ whose thread width is at most $\beta$,
    there exist a sequence of augmented global configurations $\aug{M}_0,\ldots,\aug{M}_n$, and thread mappings $\alpha$ and $\gamma$ such that $(\aug{M}_0,\ldots,\aug{M}_n,\alpha, \gamma)$ is admissible.
  \end{lemma}
  \begin{proof}
    We prove this by induction over the length of the execution.

    Base case $n=0$:
    We have an execution that only consists of a global configuration $(M_0, g_0)$ that is initial, i.e. $M_0 = \lbag \langle \body{\main}, \main, \bot, s \rangle \rbag$. Therefore we can define the corresponding augmented global configuration $\aug{M}_0 := \{ \langle \body{\main}, \main, \bot, s, \bot \rangle\}$ . It is obvious to see that $(\aug{M}_0, \alpha, \gamma)$ is admissible, without any restrictions on $\alpha$ and $\gamma$.

    Induction hypothesis:
    Assume we have a sequence of augmented global configurations $\aug{M}_0,\ldots,\aug{M}_n$ for the execution $\eta := (M_0,g_0) \semtrans{\st_1} \ldots \semtrans{\st_n} (M_n, g_n)$ whose thread width is at most $\beta$, and $\alpha$ and $\gamma$ that $(\aug{M}_0,\ldots,\aug{M}_n,\alpha,\gamma)$ is admissible.

    Induction step:
    Given the execution \[(M_0,g_0) \semtrans{\st_1} \ldots \semtrans{\st_n} (M_n, g_n) \semtrans{\st_{n+1}} (M_{n+1}, g_{n+1})\] show that there are augmented global configurations $\aug{M}_0,\dots,\aug{M}_{n+1}$ and thread mappings $\alpha, \gamma$, such that $(\aug{M}_0,\ldots,\aug{M}_n,\aug{M}_{n+1},\alpha,\gamma)$ is admissible.
    By induction hypothesis there are $\aug{M}_0,\ldots,\aug{M}_n,\alpha,\gamma$ such that $(\aug{M}_0,\ldots,\aug{M}_n,\alpha,\gamma)$ is admissible.

    Case distinction over $\st_{n+1}$:
    \begin{itemize}
        \item  $\st_{n+1} \in \atomicstmt$:
            Since $\eta$ is an execution, there are $M, X, \vartheta, t, s, X', s'$ such that:
            \begin{align*}
                M_n &= M \cup \lbag \langle X, \vartheta, t, s \rangle \rbag\\
                M_{n+1} &= M \cup \lbag \langle X', \vartheta, t, s' \rangle \rbag
            \end{align*}
            This follows from one of the rules of \srule{Assume}, \srule{Assert1}, \srule{Assert2}, \srule{AssignGlobal}, \srule{AssignLocal}, \srule{Ite1}, \srule{Ite2}, \srule{While1} or \srule{While2} combined with \srule{Frame}.
            Since $\lfloor \aug{M}_n \rfloor = M_n$ there is a $k \in \{\bot,1,\dots,\beta\}$ such that $\aug{M}_n = \aug{M} \cup \{\langle X, \vartheta, t, s, k \rangle\}$.
            Let us choose $\aug{M}_{n+1} := \aug{M} \cup \{\langle X', \vartheta, t, s', k \rangle\}$, $\alpha(n) := \langle \vartheta, k \rangle$ and $\gamma(n) := \langle \vartheta', k' \rangle$ (for arbitrary $\vartheta', k'$).

            Then $(\aug{M}_0,\ldots,\aug{M}_n,\aug{M}_{n+1},\alpha, \gamma)$ is admissible, because
            \begin{itemize}
                \item $(\aug{M}_0,\ldots,\aug{M}_n,\alpha,\gamma)$ is admissible by induction hypothesis.
                \item $\lfloor \aug{M}_{n+1} \rfloor = M_{n+1}$
                \item $\aug{M}_{n+1}$ is conformist, because $\aug{M}_n$ is conformist by induction hypothesis and the instance IDs do not change from $\aug{M}_n$ to $\aug{M}_{n+1}$.
                \item $\langle \vartheta, k, \vartheta', k' \rangle$ is $(\aug{M}_n, \aug{M}_{n+1}, \emptyset, \st_{i+1})$-awake, since $\st_{n+1} \in \atomicstmt$.
            \end{itemize}

        \item $\st_{n+1} = \texttt{fork}\;e\;\vartheta'\texttt{()}$:
        Since $\eta$ is an execution, there are $M, X, \vartheta, t, s, t', s'$ such that:
        \begin{align*}
            M_n &= M \cup \lbag \langle \texttt{fork}\;e\;\vartheta'\texttt{()};X, \vartheta, t, s \rangle \rbag\\
            M_{n+1} &= M \cup \lbag \langle X, \vartheta, t, s \rangle, \langle body_{\vartheta'}, \vartheta', t', s' \rangle \rbag
        \end{align*}
        This follows from the combination of the rules \srule{Fork} and \srule{Frame}.
        Since $\lfloor \aug{M}_n \rfloor = M_n$ there is a $k \in \{\bot,1,\dots,\beta\}$ such that
        \[\aug{M}_n = \aug{M} \cup \{\langle\texttt{fork}\;e\;\vartheta'\texttt{()};X, \vartheta, t, s, k\rangle\}\]
        Let us choose $\aug{M}_{n+1} := \aug{M} \cup \{\langle X, \vartheta, t, s, k \rangle, \langle body_{\vartheta'}, \vartheta', t', s', k' \rangle \}$ with $ k' := {\min} \{\, \tilde{k} \in \{1,\dots,\beta\} \mid
                    \lnot \exists \tilde{C},\tilde{t}, \tilde{s}\,.\, \langle\tilde{C}, \vartheta', \tilde{t}, \tilde{s},\tilde{k}\rangle \in \aug{M}_n \,\}$,  $\alpha(n) := \langle \vartheta, k \rangle$ and $\gamma(n) := \langle \vartheta', k' \rangle$. There is such a $k'$, because the thread width of $\eta$ is at most $\beta$.

        Then $(\aug{M}_0,\ldots,\aug{M}_n,\aug{M}_{n+1},\alpha, \gamma)$ is admissible, because
        \begin{itemize}
            \item $(\aug{M}_0,\ldots,\aug{M}_n,\alpha,\gamma)$ is admissible by induction hypothesis.
            \item $\lfloor \aug{M}_{n+1} \rfloor = M_{n+1}$
            \item  $\aug{M}_{n+1}$ is conformist, because $\aug{M}_n$ is conformist by induction hypothesis and the only changed instance ID from $\aug{M}_n$ to $\aug{M}_{n+1}$ is $k$, but it is a fresh ID by construction.
            \item $\langle \vartheta, k, \vartheta', k' \rangle$ is $(\aug{M}_n, \aug{M}_{n+1}, \emptyset, \st_{i+1})$-awake.
        \end{itemize}

        \item $\st_{n+1} = \texttt{join}\;e$:
        Since $\eta$ is an execution, there are $M, X, \vartheta, \vartheta', t, s, t', s'$ such that:
        \begin{align*}
            M_n &= M \cup \lbag \langle \texttt{join}\;e;X, \vartheta, t, s \rangle,  \langle \Omega, \vartheta', t', s' \rangle\rbag\\
            M_{n+1} &= M \cup \lbag \langle X, \vartheta, t, s \rangle \rbag
        \end{align*}
        This follows from the combination of the rules \srule{Join} and \srule{Frame}.
        Since $\lfloor \aug{M}_n \rfloor = M_n$ there are $k, k' \in \{\bot,1,\dots,\beta\}$ such that
        \[\aug{M}_n = \aug{M} \cup \{\langle \texttt{join}\;e;X, \vartheta, t, s, k \rangle,  \langle \Omega, \vartheta', t', s', k' \rangle\}\]
        Let us choose $\aug{M}_{n+1} := \aug{M} \cup \{\langle X, \vartheta, t, s, k \rangle\}$, $\alpha(n) := \langle \vartheta, k \rangle$ and $\gamma(n) := \langle \vartheta', k' \rangle$.

        Then $(\aug{M}_0,\ldots,\aug{M}_n,\aug{M}_{n+1},\alpha, \gamma)$ is admissible, because
        \begin{itemize}
            \item $(\aug{M}_0,\ldots,\aug{M}_n,\alpha,\gamma)$ is admissible by induction hypothesis.
            \item $\lfloor \aug{M}_{n+1} \rfloor = M_{n+1}$
            \item  $\aug{M}_{n+1}$ is conformist, because $\aug{M}_n$ is conformist by induction hypothesis and the instance IDs of $\aug{M}_{n+1}$ are only a subset of those from $\aug{M}_n$, only one thread was joined .
            \item $\langle \vartheta, k, \vartheta', k' \rangle$ is $(\aug{M}_n, \aug{M}_{n+1}, \emptyset, \st_{i+1})$-awake.
        \end{itemize}
    \end{itemize}

  \end{proof}

  Now that we have proven the existence of the augmented execution and the two thread mappings,
  we move on to the final step of our proof:
  defining the corresponding firing sequence and the sequence of states.
  Given a conformist set of augmented local configurations $\aug{M}$,
  and a set $\occ\subseteq\threadTemplates\times\{1,\ldots,\beta\}$,
  we define the marking $\marki_\occ(\aug{M})$ as follows.
  \begin{multline*}
    \marki_\occ(\aug{M}):=
      \{\, \langle C, \vartheta, k\rangle \mid \langle C, \vartheta, t, s, k\rangle\in\aug{M} \,\}\\
      \cup \{\, \inuse{\vartheta}{k}\mid \langle C, \vartheta, t, s, k \rangle\in\aug{M} \lor \langle \vartheta,k\rangle \in \occ \,\}\\
      \cup \{\, \notinuse{\vartheta}{k} \mid \lnot \exists C, t, s \,.\, \langle C,\vartheta, t, s, k \rangle \in \aug{M} \land \langle\vartheta,k\rangle \notin \occ \,\}
  \end{multline*}
  We denote by $\state{\aug{M}}{g}$ the Petri state $\sigma : \vars_\mathit{inst} \to \mathbb{Z} \cup \{\true,\false\}$
  such that $\sigma(x) = g(x)$ for all $x\in\globalVars$,
  $\sigma(\instvar{$x$}{\vartheta}{k}) = s(x)$
  and $\sigma(\idvar{\vartheta}{k}) = t$,  where $\langle C, \vartheta, t, s, k\rangle \in \aug{M}$.
  By conformism, $\state{\aug{M}}{g}$ is well-defined.

  \begin{lemma}
    \label{lem:aug-fire}
    Let $\aug{M}_1, \aug{M}_2$ be conformist sets of augmented local configurations,
    let $g_1,g_2$ be global states,
    let $\st$ be a simple statement,
    and let $\occ \subseteq \threadTemplates \times \{1,\ldots,\beta\}$.
    Given a quadruple $\langle\vartheta,k,\vartheta',k'\rangle$ that is $(\aug{M}_1,\aug{M}_2,\occ,\st)$-awake,
    if we have that $(\lfloor \aug{M}_1 \rfloor, g_1) \semtrans{\st} (\lfloor \aug{M}_2\rfloor, g_2)$,
    then it follows that
    \[
      \marki_\occ(\aug{M}_1) \fire{[\st]_{\vartheta,k}^{\vartheta',k'}} \marki_\occ(\aug{M}_2)
      \;\text{ and }\;
      \big( \state{\aug{M}_1}{g_1}, \state{\aug{M}_2}{g_2} \big) \in \sem{[\st]_{\vartheta,k}^{\vartheta',k'}}.
    \]
  \end{lemma}
  \begin{proof}
    By induction over the semantic rule from which $(\lfloor \aug{M}_1 \rfloor, g_1) \semtrans{\st} (\lfloor \aug{M}_2\rfloor, g_2)$ is derived.
    \begin{itemize}
      \item If the rule is any of
        \srule{AssignGlobal}, \srule{AssignLocal}, \srule{Assume}, \srule{Ite1}, \srule{Ite2}, \srule{While1}, \srule{While2}, \srule{Assert1}, \srule{Assert2},
        we know that $\aug{M}_1 = \{ \langle C, \vartheta, t, s, k \rangle\}$
        and $\aug{M}_2 = \{ \langle X, \vartheta, t, s ', k'\rangle \}$ for some $C, X, t, s, s', k, k'$.
        Further, for all these rules, we have that $\st_n$ is an atomic statement.
        By awakeness, it follows that $k=k'$.
        Then it follows that $\marki_\occ(\aug{M}_1)$ contains the place $\langle C,\vartheta,k\rangle$,
        and $\marki_\occ(\aug{M}_2) = \marki_\occ(\aug{M}_1) \setminus \{\langle C,\vartheta,k\rangle\} \cup \{\langle X,\vartheta,k\rangle\}$.
        Thus we have, according to \cref{eq:petrify-local}, that $\marki_\occ(\aug{M}_1) \fire{[\st]_{\vartheta,k}} \marki_\occ(\aug{M}_2)$.
        Furthermore, it is easy to see that $\big( \state{\aug{M}_1}{g_1}, \state{\aug{M}_2}{g_2} \big) \in \sem{[\st]_{\vartheta,k}}$
        for each of the above semantic rules.

      \item If the rule is \srule{Fork},
        and taking into account awakeness,
        we have
        \begin{align*}
          \aug{M}_1 &= \{ \langle \texttt{fork\;$e$\;$\vartheta'$()};X, \vartheta, t, s, k \rangle \}\\
          \aug{M}_2 &= \{\langle X, \vartheta, t, s, k \rangle, \langle \body{\vartheta'}, \vartheta', \sem{e}^{s\cup g_1}, s', k'\rangle \}
        \end{align*}
        for some $X,e,t,t',s,s'$.

        It follows that $\marki_\occ(\aug{M}_1)$ contains the place $\langle\texttt{fork\;$e$\;$\vartheta'$()};X, \vartheta, k\rangle$,
        as well as the places $\inuse{\vartheta'}{\tilde{k}}$ for all $\tilde{k} < k'$ and $\notinuse{\vartheta'}{k'}$.
        Furthermore, awakeness implies that
        \begin{multline*}
          \marki_\occ(\aug{M}_2) = \marki_\occ(\aug{M}_1) \setminus \{ \langle\texttt{fork\;$e$\;$\vartheta'$()};X, \vartheta, k\rangle, \notinuse{\vartheta'}{k'} \} \\
          \cup \{ \langle X, \vartheta, k \rangle, \langle\body{\vartheta'},\vartheta',k'\rangle, \inuse{\vartheta'}{k'} \}
        \end{multline*}
        It holds that $[\st]_{\vartheta,k}^{\vartheta',k'}$ is the statement \texttt{\idvar{\vartheta'}{k'}:=\instvar{$e$}{\vartheta}{k}}.
        According to \cref{eq:petrify-fork}, we thus have $\marki_\occ(\aug{M}_1) \fire{[\st]_{\vartheta,k}^{\vartheta',k'}} \marki_\occ(\aug{M}_2)$.
        Furthermore, it is easy to see that $\big( \state{\aug{M}_1}{g_1}, \state{\aug{M}_2}{g_2} \big) \in \sem{[\st]_{\vartheta,k}^{\vartheta',k'}}$.
      \item If the rule is \srule{Join},
        and taking into account awakeness,
        we have
        \begin{align*}
          \aug{M}_1 &= \{ \langle \texttt{join\;$e$};X, \vartheta, t, s, k \rangle, \langle \Omega,\vartheta', \sem{e}^{s\cup g_1}, s', k' \rangle \}\\
          \aug{M}_2 &= \{ \langle X, \vartheta, t, s, k \rangle \}
        \end{align*}
        for some $X,e,t,s,s'$.

        It follows that $\marki_\occ(\aug{M}_1)$ contains the places $\langle \texttt{join\;$e$};X, \vartheta, k \rangle$,
        $\langle \Omega,\vartheta', k' \rangle$ and $\inuse{\vartheta'}{k'}$.
        Furthermore, awakeness implies that
        \[
          \marki_\occ(\aug{M}_2) = \marki_\occ(\aug{M}_1) \setminus \{ \langle \texttt{join\;$e$};X, \vartheta, k \rangle, \langle\Omega,\vartheta', k'\rangle, \} \cup \{ \langle X, \vartheta, k \rangle\}
        \]
        It holds that $[\st]_{\vartheta,k}^{\vartheta',k'}$ is the statement \texttt{assume\;\idvar{\vartheta'}{k'}==\instvar{$e$}{\vartheta}{k}}.
        According to \cref{eq:petrify-join}, we thus have $\marki_\occ(\aug{M}_1) \fire{[\st]_{\vartheta,k}^{\vartheta',k'}} \marki_\occ(\aug{M}_2)$.
        Furthermore, it is easy to see that $\big( \state{\aug{M}_1}{g_1}, \state{\aug{M}_2}{g_2} \big) \in \sem{[\st]_{\vartheta,k}^{\vartheta',k'}}$.
      \item If the rule is \srule{Frame}, let $\aug{N},\aug{M}_1',\aug{M}_2'$ be sets of augmented local configurations,
        such that $\aug{M}_1 = \aug{M}_1' \uplus \aug{N}$, $\aug{M}_2 = \aug{M}_2' \uplus \aug{N}$,
        and $(\lfloor \aug{M}_1' \rfloor, g_1) \semtrans{\st} (\lfloor \aug{M}_2' \rfloor, g_2)$.
        Subsets of a conformist set are always conformist,
        and if we set $\occ' := \occ \cup \{\langle\hat\vartheta,\hat{k}\rangle \mid \langle C,\hat\vartheta,t,s,\hat{k}\in\aug{N}\}$,
        then $\langle \vartheta, k \rangle$ is $(\aug{M}_1', \aug{M}_2', \occ', \st)$-awake.
        Thus we know inductively that
        $\marki_{\occ'}(\aug{M}_1') \fire{[\st]_{\vartheta,k}} \marki_{\occ'}(\aug{M}_2')$
        and $\big( \state{\aug{M}_1'}{g_1}, \state{\aug{M}_2'}{g_2} \big) \in \sem{[\st]_{\vartheta,k}}$.
        Note that
        \begin{align*}
          \marki_\occ(\aug{M}_1) &= \marki_{\occ'}(\aug{M}_1') \cup \{\,\langle C,\vartheta',k'\rangle \mid \langle C, \vartheta',t,s,k\rangle \in \aug{N}\ ,\}\\
          \marki_\occ(\aug{M}_2) &= \marki_{\occ'}(\aug{M}_2') \cup \{\,\langle C,\vartheta',k'\rangle \mid \langle C, \vartheta',t,s,k\rangle \in \aug{N}\ ,\}
        \end{align*}
        Since the same places are added on both sides of the firing relation, we have $\marki_\occ(\aug{M}_1) \fire{[\st]_{\vartheta,k}^{\vartheta',k'}} \marki_\occ(\aug{M}_2)$.

        Finally, we observe that $\state{\aug{M}_1'}{g_1}$ and $\state{\aug{M}_1}{g_1}$ coincide on all global variables,
        as well as all instantiated variables $\instvar{$x$}{\vartheta}{k}$ and all variables \idvar{\vartheta}{k}
        such that $\langle C,\vartheta,t,s,k\rangle \in \aug{M}_1$ for any $C,t,s$.
        The analogous holds for $\state{\aug{M}_2'}{g_2}$ and $\state{\aug{M}_2}{g_2}$.
        Furthermore, by examining the control flow relation $\petrirel{\cdot}$,
        it is easy to see that $\st$ can only refer to such variables.
        Thus we conclude that $\big(\state{\aug{M}_1}{g_1}, \state{\aug{M}_2}{g_2}\big)  \in \sem{\st}$.
    \end{itemize}
  \end{proof}

\end{toappendix}
\begin{lemmarep}\label{lem:exec-fire}
  Let $(M_0,g_0) \semtrans{\st_1} \ldots \semtrans{\st_n} (M_n, g_n)$ be an execution whose thread width is at most $\beta$. %\todo{can we define this without having the thread template names in $M_i$?}.
  Then there exists a firing sequence $m_0 \fire{\tilde{\st_1}} \ldots \fire{\tilde{\st_n}} m_n$
  of $\petriProg{\beta}{\prog}$,
  and a sequence $\sigma_0,\ldots,\sigma_n$ of states over $\vars_\mathrm{inst}$
  such that
  \begin{itemize}
    %\item $\lfloor \tilde{\st_i} \rfloor = \st_i$ for all $i\in\{1,\ldots,n\}$;
    \item $\conf{m_i}{\sigma_i} = (M_i, g_i)$ for all $i\in\{0,\ldots,n\}$;
    \item and $(\sigma_{i-1},\sigma_i)\in\sem{\tilde{\st_i}}$ for all $i\in\{1,\ldots,n\}$.
  \end{itemize}
\end{lemmarep}
\begin{proofsketch}
  We first prove inductively that we can assign instance IDs to the local configurations in each step of the execution in a consistent manner.
  Given such an ``augmented'' execution,
  one can extract the markings $m_0,\ldots,m_n$ and the states $\sigma_0,\ldots,\sigma_n$ in a straightforward way.
%
%
%  where the multisets $M_i$ of local configurations are replaced by sets $\aug{M}_i$ of augmented local configurations (i.e., local configurations with a thread ID),
%  we can then define
%  \ldots

%\todo[inline]{must be more abstract, augmented global configs are not defined in the main paper}
%  Show first that there exists a sequence of extended global configurations $(\aug{M}_i, g_i)$ with an admissible active thread mapping $\alpha$.
%  Then, by induction over the length of the sequence of extended global configurations,
%  show that
%  \[
%  \marki(\aug{M}_0) \fire{[\st_1]_{\alpha(1)}} \ldots \fire{[\st_n]_{\alpha(n)}} \marki(\aug{M}_n)
%  \]
%  is a firing sequence, and that
%  \[
%  (\insta(\aug{M}_{i-1}), \insta(\aug{M}_{i})) \in \sem{[\st_i]_{\alpha(i)}}
%  \]
%  holds for all $i\in\{1,\ldots,n\}$.
\end{proofsketch}
\begin{proof}
  Let $\aug{M}_0, \ldots, \aug{M}_n$ and $\alpha,\gamma$ be admissible for the given execution.
  Such a sequence and mappings $\alpha,\gamma$ always exist by \cref{lem:exists-admissible}.
  We use them to construct the firing sequence.
  Specifically, we show that
  \[
    \marki_\emptyset(\aug{M}_0) \fire{[\st_1]_{\alpha(1)}^{\gamma(1)}} \ldots \fire{[\st_n]_{\alpha(n)}^{\gamma(n)}} \marki_\emptyset(\aug{M}_n)
  \]
  is a firing sequence, and $\big(\state{\aug{M}_{i-1}}{g_{i-1}}, \state{\aug{M}_i}{g_i}\big)  \in \sem{[\st_i]_{\alpha(i)}^{\gamma(i)}}$ for all $i\in\{1,\ldots,n\}$.

  First, we show that $\marki_\emptyset(\aug{M}_0)$ is the initial marking.
  Since $(M_0, g_0)$ is initial, we have $M_0 = \lbag \langle \body{\main},\main,\bot,s \rangle\rbag$ for some local state $s$.
  By admissibility, $\lfloor \aug{M}_0 \rfloor = M_0$, and furthermore $\aug{M}_0 = \{\langle C,\vartheta,t,s,\bot\rangle\}$ for some $C,\vartheta,t,s$.
  It follows that $\aug{M}_0 = \{\langle \body{\main},\main,\bot,s,\bot\rangle\}$,
  and hence $\marki_\emptyset(\aug{M}_0) = m_\init$.

  Second, for any $i\in\{1,\ldots,n\}$ we must show that $\marki_\emptyset(\aug{M}_{i-1}) \fire{[\st_{i}]_{\alpha(i)}^{\gamma(i)}} \marki_\emptyset(\aug{M}_{i})$
  and that $\big(\state{\aug{M}_{i-1}}{g_{i-1}}, \state{\aug{M}_i}{g_i}\big)  \in \sem{[\st_i]_{\alpha(i)}^{\gamma(i)}}$.
  But this directly follows from \cref{lem:aug-fire}.
\end{proof}

\begin{toappendix}
\subsection{Proof of \cref{thm:thread-bound-iff}}

\begin{lemma}
  Let  $m_0 \fire{\st_1} \ldots \fire{\st_n} m_n$ be a firing sequence of $\petriProg{\beta}{\prog}$,
  and let $\sigma_0,\ldots,\sigma_n$ be a sequence of instantiated states
  such that $(\sigma_{i-1},\sigma_i)\in\sem{\st_i}$ for all $i\in\{1,\ldots,n\}$
  and $m_n\cap\spec_\safe \neq \emptyset$.

  Then there always exists a firing sequence
  $m'_0 \fire{\st'_1} \ldots \fire{\st'_{n'}} m'_{n'}$ of $\petriProg{\beta}{\prog}$,
  and a sequence of instantiated states $\sigma'_0,\ldots,\sigma'_{n'}$
  such that $(\sigma'_{i-1},\sigma'_i)\in\sem{\st'_i}$ for all $i\in\{1,\ldots,n'\}$
  and $m'_n \cap \spec_\safe\neq\emptyset$,
  \emph{and} $\insuff{\vartheta} \notin m'_i$ for all $i,\vartheta$.
\end{lemma}
\begin{proof}
  If none of the $m_i$ contains a place $\insuff{\vartheta}$, we are done.

  Otherwise, let $j > 0$ be the first index such that $\insuff{\vartheta'} \in m_j$.
  Then the firing $m_{j-1} \fire{\st_j} m_j$ must be from \cref{eq:petrify-insuff},
  and thus $\st_j = \texttt{assume true}$.
  It follows that $\sigma_j = \sigma_{j-1}$,
  and $m_j = m_{j-1} \setminus \{ \inuse{\vartheta'}{1}, \ldots, \inuse{\vartheta'}{\beta}, \langle \texttt{fork\;$e$\;$\vartheta'$()};X, \vartheta, k \rangle \} \cup \{ \insuff{\vartheta'} \}$
  for some $e,X,\vartheta,k$.
  Since places $\insuff{\vartheta'}$ have no outgoing transitions,
  any transition enabled in $m_j$ is already enabled in $m_{j-1}$.
  Further, $\sigma_j=\sigma_{j-1}$.
  Hence we can omit the $j$-th transition, and omit $\sigma_j$ from the state sequence.
  We arrive at a firing sequence and a sequence of states as described.
\end{proof}
\end{toappendix}

%Using these two lemmata,
%we show the three main results on petrification:
%
%
\begin{theoremrep}[Thread Width Detection]%
  \label{thm:thread-bound-iff}%
  The Petri program $\petriProg{\beta}{\prog}$ satisfies the bound specification $\spec_\bound$
  iff
  the thread width for $\prog$ is at most $\beta$.
\end{theoremrep}
\begin{proofsketch}
  Given an execution with a thread width greater than $\beta$,
  we apply \cref{lem:exec-fire} to the longest prefix of the execution
  such that the thread width of the prefix is at most $\beta$.
  The firing sequence and the sequence of states given by the lemma can be extended to reach a place $\insuff{\vartheta}$.
  For the reverse implication, we proceed analogously
  by applying \cref{lem:fire-exec} to the longest prefix of a given firing sequence that does not put a token into a place $\insuff{\vartheta}$.
  The resulting execution can be extended to an execution with thread width greater than $\beta$.
\end{proofsketch}
\begin{proof}
  We first show that, if the Petri program $\petriProg{\beta}{\prog}$ satisfies its bound specification,
  then the thread width for $\prog$ is at most $\beta$.
  We show this by contraposition:
  Suppose that the thread width for $\prog$ is greater than $\beta$.
  Then there would exist an execution $(M_0,g_0) \semtrans{\st_1} \ldots \semtrans{\st_n} (M_n, g_n)$
  such that some $M_i$ had strictly more than $\beta$ thread instances of some template $\vartheta'$.
  Wlog.\ we assume that the execution is minimal,
  i.e., for all $i \in \{0,\ldots,n-1\}$ we have at most $\beta$ thread instances (for all templates) in $M_i$.
  Let $m_0 \fire{\tilde{\st_1}} \ldots \fire{\tilde{\st_{n-1}}} m_{n-1}$ be the corresponding firing sequence,
  and let $\sigma_0,\ldots,\sigma_{n-1}$ the corresponding sequence of states,
  as given by \cref{lem:exec-fire}.
  In particular, we have $\conf{m_{n-1}}{\sigma_{n-1}} = (M_{n-1}, g_{n-1})$.
  The only semantic rule that increases the number of threads is \srule{Fork},
  and it increases the number by exactly 1.
  Hence we know that $M_{n-1}$ must have $\beta$ thread instances of template $\vartheta'$,
  and thus by definition of $\confOp$ and coherence of $m_{n-1}$,
  we conclude that $\inuse{\vartheta'}{1}, \ldots, \inuse{\vartheta'}{\beta} \in m_{n-1}$.
  Furthermore, since $(M_{n-1},g_{n-1}) \semtrans{\texttt{fork\;$e$\;$\vartheta'$()}} (M_n, g_n)$ (for some expression $e$),
  we know that there exists a local configuration $\langle C, \vartheta, t, s \rangle \in M_{n-1}$,
  where $C = \texttt{fork e $\vartheta'$()};X$ with $X\in\Command\cup\{\Omega\}$.
  By definition of $\confOp$, we must have $\langle C, \vartheta, k \rangle\in m_{n-1}$ for some $k$.
  It follows from \cref{eq:petrify-insuff} that
  $m_{n-1} \fire{\texttt{assume true}} m_n$,
  where $m_n := m_{n-1} \setminus \{\inuse{\vartheta'}{1}, \ldots, \inuse{\vartheta'}{\beta},\langle C, \vartheta, k \rangle\} \cup \{\insuff{\vartheta'},\langle X,\vartheta, k\rangle \}$.
  Furthermore, we set $\sigma_k := \sigma_{k-1}$.
  It follows that the firing sequence $m_0\fire{\tilde{\st}_1} \ldots \fire{\tilde{\st}_{n-1}} m_{n-1} \fire{\texttt{assume\;true}} m_n$
  and the states $\sigma_0,\ldots,\sigma_n$ form a counterexample to the bound specification $\spec_\bound$.
  Thus we have shown that $\petriProg{\beta}{\prog}$ violates its bound specification.

  For the reverse implication, we again proceed by contraposition.
  Suppose $m_0 \fire{\tilde{\st_1}} \ldots \fire{\tilde{\st_n}} m_n$ is an accepting firing sequence of $\petriProg{\beta}{\prog}$,
  and $\sigma_0,\ldots,\sigma_n$ a sequence of Petri states with $(\sigma_{i-1}, \sigma_i) \in \sem{\tilde{\st_i}}$ for all $i$.
  Wlog.\ we assume that the firing sequence is minimal, i.e., none of the markings up to $m_{n-1}$ contains the place $\insuff{\vartheta}$.
  From \cref{lem:fire-exec},
  it follows that $\conf{m_0}{\sigma_0} \semtrans{\st_1} \ldots \semtrans{\st_{n-1}} \conf{m_{n-1}}{\sigma_{k-1}}$ is an execution,
  for some statements $\st_1,\ldots,\st_n$.
  The marking $m_{n-1}$ must enable a transition according to \cref{eq:petrify-insuff},
  hence $\langle C, \vartheta, k\rangle \in m_{n-1}$ where $C = \texttt{fork e $\vartheta'$()};X$ for some $X\in\Command\cup\{\Omega\}$.
  Then there exists a local configuration $\langle C, \vartheta, t, s\rangle$ in $M_{n-1}$.
  By \srule{Fork}, we have that $(M_{n-1}, g_{n-1}) \semtrans{\texttt{fork $e$ $\vartheta'$()}} (M_n, g_n)$,
  where $g_n := g_{n-1}$ and $M_n := M_{n-1} \setminus \lbag \langle C, \vartheta, t, s\rangle \rbag \uplus \lbag \langle X, \vartheta, t, s\rangle, \langle\body{\vartheta}, \sem{e}^{s\cup g_{n-1}}, s' \rangle \rbag$.
  Furthermore, since $m_{n-1}$ enables a transition according to \cref{eq:petrify-insuff},
  we must have $\inuse{\vartheta'}{1}, \ldots, \inuse{\vartheta'}{\beta} \in m_{n-1}$.
  By coherence and definition of $\confOp$,
  there must already exist $\beta$ local configurations for the template $\vartheta'$ in $M_{n-1}$.
  Hence $M_n$ has more than $\beta$ thread instances,
  and thus the the thread width of the execution $(M_0,g_0) \semtrans{\st_1} \ldots \semtrans{\st_{n-1}} (M_{n-1}, g_{n-1}) \semtrans{\texttt{fork\;$e$\;$\vartheta'$()}} (M_n,g_n)$ is greater than $\beta$.
\end{proof}

\begin{theorem}[Soundness]
  \label{thm:soundness}
  If the Petri programs $\petriProg{\beta}{\prog}$ satisfies both its safety and its bound specifications,
  then the \proglang program $\prog$ is correct.
\end{theorem}
\begin{proof}
  From the fact that $\petriProg{\beta}{\prog}$ satisfies the bound specification, we conclude by \cref{thm:thread-bound-iff}
  that the thread width for $\prog$ is at most $\beta$.
  Contrapositively, we prove that if $\prog$ is incorrect, then $\petriProg{\beta}{\prog}$ does not satisfy its safety specification.
  Suppose that $(M_0, g_0) \semtrans{\st_1} \ldots \semtrans{\st_n} (M_n, g_n)$ is an erroneous execution.
  \Cref{lem:exec-fire} gives us a corresponding firing sequence $m_0 \fire{\tilde{\st_1}} \ldots \fire{\tilde{\st_n}} m_k$
  and a sequence of states $\sigma_0,\ldots,\sigma_k$.
  Since $M_n$ contains some local configuration $\langle \lightning, \vartheta, t, s \rangle$,
  by definition of $\confOp$ we must have a place $\langle \lightning, \vartheta, k \rangle \in m_n$.
  Thus the firing sequence and the sequence of states form a counterexample to the safety specification $\spec_\safe$.
\end{proof}

\begin{theorem}[Completeness]
  \label{thm:completeness}
  If the \proglang program $\prog$ is correct,
  then the corresponding Petri program $\petriProg{\beta}{\prog}$ satisfies its safety specification.
\end{theorem}
\begin{proof}
  Contrapositively, let us suppose that $\petriProg{\beta}{\prog}$ does not satisfy its safety specification.
  Then there exists a firing sequence $m_0 \fire{\tilde{\st_1}} \ldots \fire{\tilde{\st_n}} m_n$
  and states $\sigma_0,\ldots,\sigma_n$ such that $(\sigma_{i-1}, \sigma_i)\in\sem{\tilde{\st_i}}$ for all $i$,
  with some $\langle \lightning, \vartheta, k \rangle \in m_n$.
  Wlog.\ we can assume that the firing sequence does not run into a place $\insuff{\vartheta'}$.
  By \cref{lem:fire-exec},
  we know that $\conf{m_0}{\sigma_k} \semtrans{\st_1} \ldots \semtrans{\st_n} \conf{m_n}{\sigma_n}$
  is an execution, for some $\st_1,\ldots,\st_n$.
  %But % since the firing sequence is accepting ($\langle \lightning, \vartheta, j \rangle \in m_k$ for some $\vartheta, j$),
  By definition of $\confOp$,
  it follows that the execution is erroneous,
  i.e., $\langle \lightning, \vartheta, t, s \rangle \in \conf{m_n}{\sigma_n}$.
  Thus $\prog$ is incorrect.
\end{proof}

\begin{toappendix}
\begin{figure}[b]
\begin{align*}
  \langle \texttt{assume e}; X,\vartheta, k \rangle &\petrirel{\texttt{assume e}} \langle X,\vartheta, k \rangle\\[0.5em]
%  \langle \texttt{havoc x},\vartheta, k \rangle &\petrirel{\texttt{havoc x}} \langle \Omega, \vartheta,k \rangle\\[0.5em]
  \langle \texttt{x:=e}; X,\vartheta, k \rangle &\petrirel{\texttt{x:=e}} \langle X, \vartheta,k \rangle\\[1em]
  \langle \texttt{assert e}; X, \vartheta, k \rangle &\petrirel{\texttt{assume e}}  \langle X, \vartheta, k \rangle\\[0.5em]
  \langle \texttt{assert e}; X, \vartheta, k \rangle &\petrirel{\texttt{assume !e}} \langle \lightning, \vartheta, k \rangle\\[1em]
  \langle \texttt{if\,($e$)\,\{\,$C_1$\,\}\,else\,\{\,$C_2$\,\}}; X, \vartheta, k \rangle &\petrirel{\texttt{assume e}}  \langle C_1; X,\vartheta, k \rangle\\[0.5em]
  \langle \texttt{if\,($e$)\,\{\,$C_1$\,\}\,else\,\{\,$C_2$\,\}}; X, \vartheta, k \rangle &\petrirel{\texttt{assume !e}} \langle C_2; X,\vartheta, k \rangle\\[1em]
  \langle \texttt{while\,($e$)\,\{\,$C$\,\}}; X, \vartheta, k \rangle &\petrirel{\texttt{assume e}}  \langle \texttt{$C$\,;\,while\,($e$)\,\{\,$C$\,\}}; X,\vartheta, k \rangle\\[0.5em]
  \langle \texttt{while\,($e$)\,\{\,$C$\,\}}; X, \vartheta, k \rangle &\petrirel{\texttt{assume !e}} \langle X, \vartheta, k \rangle\\[1em]
  \inuse{\vartheta}{1},\ldots,\inuse{\vartheta}{k-1},\notinuse{\vartheta}{k},\langle \texttt{fork e $\vartheta$()}; X,\vartheta', k' \rangle &\petrirel{\texttt{\idvar{\vartheta}{k}:=e}} \langle \body{\vartheta}, \vartheta, k \rangle, \langle X, \vartheta', k' \rangle,\inuse{\vartheta}{1},\ldots,\inuse{\vartheta}{k}% \qquad \forall j \in \{1,\ldots,n\}
  \\[0.5em]
  \inuse{\vartheta}{1},\ldots,\inuse{\vartheta}{\beta},\langle \texttt{fork e $\vartheta$()}; X, \vartheta', k \rangle &\petrirel{\texttt{assume true}} \insuff{\vartheta}\\[1em]
  \inuse{\vartheta}{j}, \langle \Omega, \vartheta,k \rangle, \langle \texttt{join e}; X, \vartheta', k' \rangle &\petrirel{\texttt{assume \idvar{\vartheta}{k}==e}} \langle X, \vartheta', k' \rangle, \notinuse{\vartheta}{k}% \qquad \forall j \in \{1,\ldots,n\}
  \\[1em]
  %
%  \infer{\langle C_1;C_2, \vartheta, i \rangle \uplus Z_1 \petrirel{\st} \langle C_1';C_2, \vartheta,i \rangle \uplus Z_2}{
%    \langle C_1, \vartheta,i \rangle \uplus Z_1 \petrirel{\st} \langle C_1', \vartheta,i\rangle \uplus Z_2
%  }\\[0.5em]
%  %
%  \infer{\langle C_1;C_2, \vartheta, i \rangle \uplus Z_1 \petrirel{\st} \langle C_2, \vartheta, i \rangle \uplus Z_2}{
%      \langle C_1, \vartheta, i \rangle \uplus Z_1 \petrirel{\st} \langle \Omega, \vartheta, i\rangle \uplus Z_2
%    }\\
%  %
%  \infer{\langle C_1;C_2, \vartheta, i \rangle \uplus Z_1 \petrirel{\st} \langle \lightning, \vartheta, i \rangle \uplus Z_2}{
%      \langle C_1, \vartheta,i \rangle \uplus Z_1 \petrirel{\st} \langle \lightning, \vartheta, i\rangle \uplus Z_2
%    }
%  %
\end{align*}
\caption{The full definition of the Petri transition relation $\petrirel{\cdot}$.}
\label{fig:full-petrification-trans}
\end{figure}
%\todo[inline]{
%  unify syntax and caption of \cref{fig:full-petrification-trans} with rest of paper
%}
\end{toappendix}

\section{Verifying Programs through Repeated Petrification}
\label{sec:verification}
The previous section shows that one can verify a $\proglang$ program $\prog$
by picking a suitable thread limit $\beta$,
and proving that the petrified program $\petriProg{\beta}{\prog}$ satisfies both its safety and its bound specification.
This gives rise to several possible verification algorithms,
illustrated in \cref{fig:algorithms}.
%, which we present and evaluate in this section.
%
%\subsection{Verification Algorithms}
%
\begin{figure}[t]
\begin{subfigure}{0.48\textwidth}
\begin{tikzpicture}[thick,baseline]
  \node[draw] (safe?) {$\petriProg{\beta}{\prog} \models \spec_\mathsf{safe}$ ?};
  \node[draw,below=1cm of safe?] (bound?) {$\petriProg{\beta}{\prog} \models \spec_\mathsf{bound}$ ?};
  \node[right=1cm of safe?] (incorrect) {incorrect};
  \node[right=1cm of bound?] (correct) {correct};

  \draw[<-] (safe?) -- node[auto]{$\beta \gets 1$} ++(up:1cm);
  \draw[->] (safe?) -- node[auto]{\textbf{no}} (incorrect);
  \draw[->] (safe?) -- node[auto,swap]{\textbf{yes}} (bound?);
  \draw[->] (bound?) -- node[auto]{\textbf{yes}} (correct);
  \draw[->] (bound?) edge[bend right=35] node[auto,swap]{\textbf{no}, $\beta\gets \beta+1$} (safe?);
\end{tikzpicture}
\hfill
\caption{Algorithm 1}
\end{subfigure}
\begin{subfigure}{0.48\textwidth}
\hfill
\begin{tikzpicture}[thick,baseline]
  \node[draw] (bound?) {$\petriProg{\beta}{\prog} \models \spec_\mathsf{bound}$ ?};
  \node[draw,below=1cm of bound?] (safe?) {$\petriProg{\beta}{\prog} \models \spec_\mathsf{safe}$ ?};
  \node[below=0.5cm of safe?,xshift=-0.75cm] (incorrect) {incorrect};
  \node[below=0.5cm of safe?,xshift=0.75cm] (correct) {correct};

  \draw[<-] (bound?) -- node[auto]{$\beta \gets 1$} ++(up:1cm);
  \draw[->] (bound?) edge[out=-30,in=30,looseness=8] node[auto,swap,align=left]{\textbf{no},\\$\beta\gets \beta+1$} (bound?);
  \draw[->] (safe?.south) ++(left:0.75cm) -- node[auto,swap]{\textbf{no}} (incorrect);
  \draw[->] (bound?) -- node[auto,swap]{\textbf{yes}} (safe?);
  \draw[->] (safe?.south) ++(right:0.75cm) -- node[auto]{\textbf{yes}} (correct);
\end{tikzpicture}
\caption{Algorithm 2}
\end{subfigure}
\begin{subfigure}{\textwidth}
\vspace{-5mm}
\begin{tikzpicture}[thick,baseline]
  \node[draw] (check?) {$\petriProg{\beta}{\prog} \models \spec_\mathsf{safe} \cup \spec_\mathsf{bound}$ ?};
  \node[below=1cm of check?,xshift=-2.5cm] (incorrect) {incorrect};
  \node[below=1cm of check?,xshift=2.5cm] (correct) {correct};

  \draw[<-] (check?) -- node[auto]{$\beta \gets 1$} ++(up:1cm);
  \draw[->] (check?) edge[loop below] node[auto,align=center]{\textbf{no:}\\\textbf{ctex violates $\spec_\mathsf{bound}$},\\ $\beta\gets \beta+1$} ();
  \draw[->] (check?) -| node[auto,swap,align=right]{\textbf{no:}\\\textbf{ctex violates $\spec_\mathsf{safe}$}} (incorrect);
  \draw[->] (check?) -| node[auto]{\textbf{yes}} (correct);
\end{tikzpicture}
\caption{Algorithm 3}
\label{fig:algo3}
\end{subfigure}
  \caption{Three iterative algorithms that reduce the verification problem for a \proglang program $\prog$ to (several instances of) the verification problem for Petri programs.}
  \label{fig:algorithms}
\end{figure}
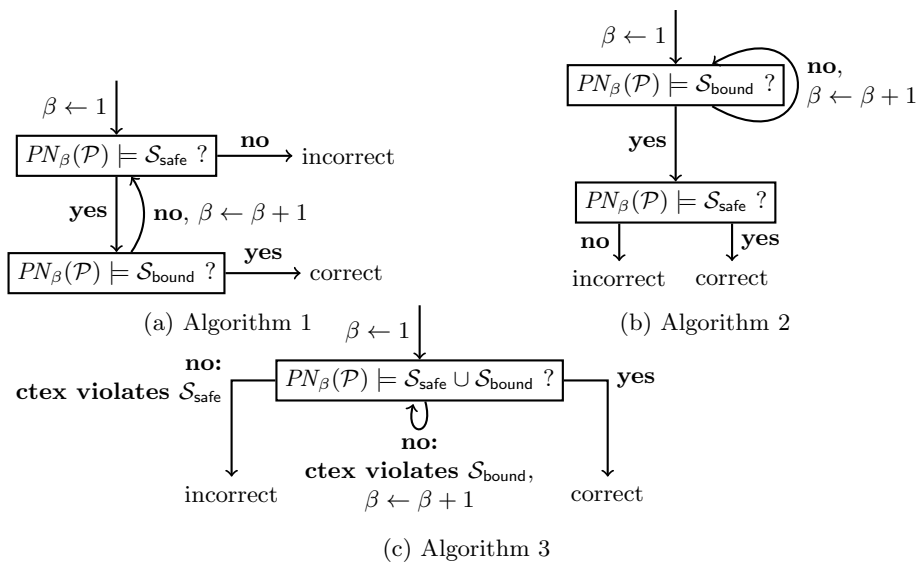
%
%In \cref{fig:algorithms}, we illustrate three possible verification algorithms based on the insights from the previous section.
Each algorithm proves correctness of a given program $\prog$
by iteratively determining a suitable thread limit $\beta$.
In each iteration, the algorithms invoke a Petri program verification algorithm~\cite{vmcai2021}
to determine if the petrified program $\petriProg{\beta}{\prog}$ satisfies some given specification.

Algorithm~1 first checks safety,
and only if the petrified program satisfies the safety specification,
the algorithm checks the bound specification.
If the bound specification is violated, the thread limit $\beta$ is increased.
The safety specification must then be checked again, as it may be violated for the increased $\beta$.

By contrast, Algorithm~2 first determines the thread width $\beta$ of the program,
by repeatedly checking if the petrified program satisfies the bound specification
and, if appropriate, incrementing the thread limit.
Only after the thread width has been established, the petrified program is checked against the safety specification.

Finally, Algorithm~3 combines the check of both specifications,
by checking if the petrified program satisfies their union.
If this is not the case, the counterexample returned by the verification is examined to determine if the program is incorrect,
or if the thread limit should be increased.

%In all three cases, it is straightforward to see:
%
\begin{theorem}[Correctness]
  For any of the algorithms in \cref{fig:algorithms},
  if the algorithm terminates for a given program $\prog$,
  then
  \begin{itemize}
  \item if the output is ``correct'', the program $\prog$ is correct;
  \item and if the output is ``incorrect'', the program $\prog$ is incorrect.
  \end{itemize}
\end{theorem}
\begin{proof}
  The result is straightforward by application of \cref{thm:thread-bound-iff,thm:soundness,thm:completeness}.
  For the third algorithm, we also note that a Petri program satisfies the union of two specifications iff
  it satisfies both specifications individually.
\end{proof}

%The verification may fail to terminate, either if one of the Petri program verification problems can not be solved (i.e., some invocation of the algorithm in \cite{vmcai2021} fails to terminate),
%or if $\prog$ has an infinite thread width. % $\beta$ for $\prog$. % (i.e., $\prog$ is \emph{unbounded}) is not bounded.
%
\begin{theorem}[Relative Termination]
  Given a program $\prog$ with a finite thread width $\beta$,
  if each invocation of the Petri program verification algorithm~\cite{vmcai2021} terminates,
  then all the verification algorithms in \cref{fig:algorithms}
  terminate after at most $\beta$ iterations.
\end{theorem}
The first algorithm may terminate earlier if $\prog$ is incorrect,
and in fact, it even terminates if $\prog$ is incorrect but has infinite thread width:
If an erroneous execution exists, this execution has some thread width $\beta$.
Thus, in iteration $\beta$ at the latest, the petrified program $\petriProg{\beta}{\prog}$ does not satisfy the safety specification $\spec_\mathsf{safe}$,
and Algorithm~1 terminates.
Algorithm~2 never terminates for programs with infinite thread width.
For Algorithm~3, termination is not guaranteed and depends on the counterexample selection of the underlying Petri program verification algorithm.

%In the next section, we compare the three algorithms empirically.

% !TEX root = ../main.tex

\section{Application: Verification of C Programs}
\label{sec:implementation}
We implemented petrification, as well as the three verification algorithms of \cref{fig:algorithms},
in the program analysis framework \toolname~\cite{ultimate-website}.
Our implementation consumes C programs that use the POSIX threads (pthreads) API~\cite{pthreads} for dynamic thread management. We presume an execution model that satisfies sequential consistency.
%\todo{Clarify difference between Ultimate and Automizer/Taipan/GemCutter}

\subsection{Translation from pthreads to \proglang Programs}
%
%Let us explain how C programs that use the pthreads API are translated to \proglang programs.
%\toolname translates
%Here we focus on the translation of the pthread features and skip over the remaining C features, since they are out-of-scope.
The POSIX threads extension includes many features, including mutexes and condition variables. Our implementation supports a subset of these features through a symbolic encoding in \proglang using assume statements.
Here, we focus on the features relevant to dynamic thread management, specifically:
\begin{itemize}
    \item There are unique thread IDs of type \texttt{pthread\_t}.\medskip

    In  \proglang, to ensure unique thread IDs,
    we introduce a global integer variable \texttt{freshId} that is incremented after every \texttt{fork}.\medskip
    \item The function \texttt{pthread\_create(id, attr, f, arg)} creates a new thread.
    It takes a pointer \texttt{id} to \texttt{pthread\_t} (the unique thread ID will be stored at this address),
    some attributes \texttt{attr}, %\todo{for us must be NULL}
    the function \texttt{f} which serves as thread template,
    and an argument \texttt{arg} (a pointer) passed to the function \texttt{f}.
    \medskip

    In \proglang, a call to \texttt{pthread\_create(id, attr, f, arg)} is translated to three statements.
          First, we write the current value \texttt{freshId} to the location of the pointer \texttt{id}.
          Second, we execute the statement \texttt{fork freshId f(arg)}.
          Finally, we increment \texttt{freshId}.

          Our implementation allows thread templates (created from C functions) to take parameters,
                and extends the syntax and semantics of \proglang  to allow
     passing such parameters in \texttt{fork} statements.
          %      (\texttt{arg} is omitted if \texttt{NULL}) \todo{Is this correct and how to explain it here?}.
          %
          %
          We do not support non-standard attributes, so \texttt{attr} must be \texttt{NULL}.\medskip

    %
    %It works similar to the \texttt{fork} in \proglang, after it is executed the thread is active.
    %Additionally a unique thread ID is written to \texttt{id}.
    %
    \item The function \texttt{pthread\_join(id, ret)} waits for the thread with the given ID to terminate.
    It takes the thread ID \texttt{id} (of type \texttt{pthread\_t}) and a pointer \texttt{ret} to store the return value of the function to be joined
    (if \texttt{ret} is \texttt{NULL}, the return value is discarded).
    %
    %It works similar to the \texttt{join} in \proglang, after it is executed, we wait for the thread with the thread id to terminate.
    \medskip

    In \proglang, a call to \texttt{pthread\_join(id, ret)} is translated to the \proglang statement \texttt{join~id}.
          To cover the case that \texttt{ret} is not \texttt{NULL},
          our implementation supports  statements of the form \texttt{join id assigns x},
          which store the thread's return value in the variable \texttt{x}.
\end{itemize}

\begin{figure}[t]
\begin{minipage}[t]{0.585\textwidth}
\begin{lstlisting}
int c, i;


void *w(void *x) {
  c += i;
  assert(c <= 2 * i);
  c -= i;
}


int main() {
  pthread_t ids[10000];
  while (i < 10000) {
    pthread_create(&ids[i], NULL, w, NULL);
    if (i > 0) {
      pthread_join(ids[i-1], NULL);
    }
    i++;
  }
  return 0;
}\end{lstlisting}
\vspace{2.3mm}
\caption{C program using the pthreads API}
\label{fig:example_pthread}
\end{minipage}
\hfill
\begin{minipage}[t]{0.394\textwidth}
    \hfill\begin{subfigure}[t]{0.95\textwidth}
\begin{lstlisting}
c := 0;
i := 0;
while (i < 10000) {
  ids[i] := freshId;
  fork freshId w();
  freshId := freshId + 1;
  if (i > 0) {
    join ids[i-1];
  }
  i := i + 1;
}
\end{lstlisting}
\vspace{-2mm}
        \caption{The \texttt{main} thread}
    \end{subfigure}\\

    \hfill\begin{subfigure}[t]{0.95\textwidth}
\begin{lstlisting}
c := c + i;
assert c <= 2 * i;
c := c - i;
\end{lstlisting}
\vspace{-2mm}
        \caption{The worker thread \texttt{w}}
    \end{subfigure}
\caption{The representation of the program from \cref{fig:example_pthread} in \proglang}
\label{fig:example_pthread:translated}
\end{minipage}
\end{figure}
%
%We encode these feature in \proglang in the following ways:
%
%\begin{enumerate}
%    \item
%To ensure unique thread IDs,
%we introduce a global integer variable \texttt{freshId} that is incremented after every \texttt{fork}.
%
%    \item
%A call to \texttt{pthread\_create(id, attr, f, arg)} is translated to three statements.
%      First, we write the current value \texttt{freshId} to the location of the pointer \texttt{id}.
%      Second, we execute the statement \texttt{fork freshId f(arg)}.
%      Finally, we increment \texttt{freshId}.
%
%      Our implementation allows thread templates (created from C functions) to take parameters,
%            and extends the syntax and semantics of \proglang  to allow
% passing such parameters in \texttt{fork} statements.
%      %      (\texttt{arg} is omitted if \texttt{NULL}) \todo{Is this correct and how to explain it here?}.
%      %
%      %
%      We do not support non-standard attributes, so \texttt{attr} must be \texttt{NULL}.
%
%    \item
%A call to \texttt{pthread\_join(id, ret)} is translated to the \proglang statement \texttt{join~id}.
%      To cover the case that \texttt{ret} is not \texttt{NULL},
%      our implementation supports  statements of the form \texttt{join id assigns x},
%      which store the thread's return value in the variable \texttt{x}.
%\end{enumerate}
%
\Cref{fig:example_pthread} shows a version of our example program from \cref{fig:example},
written in C using the pthreads API.
%  you can find a C program that behaves in a similar way to the program from \cref{fig:example}.
The \texttt{main} thread creates 10\,000 instances of the thread~\texttt{w} using \texttt{pthread\_create},
and stores the unique thread IDs in the array \texttt{ids}. % stores the unique thread IDs.
In each iteration, the thread created in the the previous iteration is joined with \texttt{pthread\_join(ids[i-1], NULL)}.
\Cref{fig:example_pthread:translated} shows the translation of the program to \proglang as described above. %\todo{Describe it on example or is the general approach sufficient?}
Petrification of this \proglang program with thread limit $\beta=2$ yields a Petri program similar to \cref{fig:petri-net}.
\toolname can then prove that this Petri program satisfies the bound and safety specification.
Thus the C program in \cref{fig:example_pthread} is correct and has thread width $2$.
% is indeed the thread width of this program
%(there are at most two thread instances active at each time)
%and that the petrified program cannot reach an error location.
%(the assert statement is never violated)
%Thus the C program in \cref{fig:example_pthread} is correct.

\subsection{Practical Performance of the Approach}
The \toolname framework contains three verification tools that verify Petri programs:
\textsc{Automizer}~\cite{svcomp23:automizer-commuhash,vmcai2021}, \textsc{Taipan}~\cite{taipan,sas17:taipan} and \textsc{GemCutter}~\cite{pldi22:sound-seq,gemcutter}.
\toolname can encode different specifications via \texttt{assert} statements, such as unreachability of an error function, or absence of certain undefined behaviours. % in the \proglang program.
\begin{table}[b]
\vspace*{-2em}
    \centering
    \caption{Comparison of SV-COMP'23~\cite{beyer:svcomp23} results}
    \label{tbl:eval-results2}
    \setlength{\tabcolsep}{6pt}
    \begin{tabular}{l|rrr|rrr|rrr}
        & \multicolumn{3}{c|}{\textsc{Automizer}} & \multicolumn{3}{c|}{\textsc{CPAchecker}} & \multicolumn{3}{c}{\textsc{Goblint}}  \\
        &    & time & mem  &    & time & mem  &    & time & mem \\
        & \# & (h)  & (GB) & \# & (h)  & (GB) & \# & (h)  & (GB) \\
        \hline
        total (2\,865) & 1\,516 & 35.3     & 1\,590        & 973  & 16.1      & 1\,210     & 847 & 0.4      & 29          \\
        \quad safe   & 1\,227 & 30.6     & 1\,300        & 712  & 10.8      & 900      & 847 & 0.4      & 29          \\
        \quad unsafe  &  289 &  4.8     & 290         & 261  & 5.3       & 310      & 0   & 0.0      & 0
    \end{tabular}
\end{table}

The feasibility of our presented approach is demonstrated by the success of the \toolname tools in the \emph{ConcurrencySafety} category of the \emph{International Competition on Software Verification} (SV-COMP'23)~\cite{beyer:svcomp23}.
In this category, verification tools had to check 2865 verification tasks,
each consisting of a concurrent C program and one of four different specifications:
unreachability of a call to a distinguished error function,
absence of data races,
absence of invalid pointer dereferences and other memory safety issues,
and absence of signed integer overflow.
The \toolname tools participated (using Algorithm~3 as shown in \cref{fig:algo3}) and occupied the 2\textsuperscript{nd} (\textsc{Automizer}), 3\textsuperscript{rd} (\textsc{GemCutter}) and 4\textsuperscript{th} (\textsc{Taipan}) place in the \textit{ConcurrencySafety} category,
with 1\textsuperscript{st} place taken by the bounded model checker \textsc{Deagle}~\cite{deagle}.
Thus the \toolname tools placed ahead of all other tools that soundly verify concurrent programs.
\Cref{tbl:eval-results2} shows an extract of the competition results,
comparing \textsc{Automizer} against the next best sound verification tools in the \textit{ConcurrencySafety} category, \textsc{CPAchecker}~\cite{cpachecker} and \textsc{Goblint}~\cite{goblint:sv-comp}.

%\toolname competed against several mature tools and placed among the top 5 tools in the \textsc{ConcurrencySafety} category%
%\footnote{exact place withheld for anonymity}%
%.

We also evaluated \toolname's implementation of all three verification algorithms from \cref{fig:algorithms} on the SV-COMP'23 benchmark set, with \textsc{Automizer}~\cite{vmcai2021} as a backend.
This evaluation was performed using the \textsc{BenchExec} benchmarking tool~\cite{beyer:benchexec} on an AMD Ryzen Threadripper 3970X 32-Core processor, with a timeout of 15\,min and a memory limit of 16\,GB.
\Cref{tbl:eval-results} shows the results: %\todo{update! all props}
Algorithm~1 succeeds on the largest number of unsafe benchmarks and overall.
For unsafe benchmarks it succeeds on a strict superset of the verification tasks for which Algorithms~2 and 3 found a bug.
However, Algorithm~2 and Algorithm~3 are both able to verify a few more safe benchmarks.
%though the differences are minute.
%The five\todo{update!} tasks where only Algorithm~1 succeeds %but not by both others
%are unsafe programs that have either no or a very high (10\,000) thread bound.
The results are similar, likely because
% contains very few unbounded unsafe benchmarks.
most programs in the benchmark set either have thread width $1$ (but multiple thread templates),
or they are safe and have infinite thread width.

\begin{table}[t]
\centering
  \caption{Comparison of the algorithms from \cref{fig:algorithms}}
  \label{tbl:eval-results}
  \setlength{\tabcolsep}{6pt}
  \begin{tabular}{l|rrr|rrr|rrr}
    & \multicolumn{3}{c|}{Algorithm 1} & \multicolumn{3}{c|}{Algorithm 2} & \multicolumn{3}{c}{Algorithm 3}\\
    &    & time & mem  &    & time & mem  &    & time & mem \\
    & \# & (h)  & (GB) & \# & (h)  & (GB) & \# & (h)  & (GB) \\
    \hline
    total (2\,865) & 1\,580 & 26.1     & 2\,208        & 1\,571  & 26.6     & 2\,227        & 1\,573 & 26.1     & 2\,176        \\
    \quad safe   & 1\,224 & 22.2     & 1\,700        & 1\,230  & 22.9     & 1\,760        & 1\,225 & 22.3     & 1\,690        \\
    \quad unsafe & 356  & 4.0      & 508         & 341   & 3.7      & 467         & 348  & 3.8      & 486
  \end{tabular}
\end{table}

% !TEX root = ../main.tex

\section{Conclusion}
\label{sec:conclusion}
We address the verification of programs with dynamic thread management, i.e., programs that fork and join threads at runtime.
Our main contributions are:
\begin{itemize}
 \item We present a reduction from the verification of programs with dynamic thread management
 to the verification of programs with a fixed number of threads.
 Our approach determines the maximum number of threads that are active at the same time (the program's \emph{thread width})
% Our approach determines the thread width of a program,
 and verifies that the program indeed satisfies this bound.
 \item
 %that utilizes an algorithm for the verification of programs with a fixed number of threads
 %in order to analyze whether a program with dynamic thread management satisfies a given %safety specification.
 %
 We formalize our approach using Petri programs, an existing model for concurrent programs with a fixed number of threads,
 and a simple programming language that supports dynamic thread management.
 %
 %We
  %presented a translation from \proglang to Petri programs and
 %prove that
 Our approach is sound, and it is
 relatively complete for programs with a finite thread width.

% Our overall approach is suitable
% to detect the thread bound of the program, and sound and relatively complete for the verification of safety properties.
 \item We
 %presented three variants of our approach and
 implemented our approach as a %state-of-the-art
 verification tool for C programs.
 Our implementation verifies programs that use the POSIX threads (pthreads) API,
 and
 %We demonstrate that
% Our tools
 %allows that tool to match up
 competes with the best verifiers for concurrent programs at the International Competition on Software Verification (SV-COMP'23)~\cite{beyer:svcomp23}.
\end{itemize}
\goodbreak

\bibliographystyle{splncs04}
\bibliography{references}
\end{document}